\newcommand{\setA}{\mathcal{A}}
\newcommand{\setB}{\mathcal{B}}
\newcommand{\setC}{\mathcal{C}}
\newcommand{\setD}{\mathcal{D}}
\newcommand{\setI}{\mathcal{I}}
\newcommand{\setK}{\mathcal{K}}
\newcommand{\setM}{\mathcal{M}}
\newcommand{\setQ}{\mathcal{Q}}
\newcommand{\setR}{\mathcal{R}}
\newcommand{\setU}{\mathcal{U}}
\newcommand{\setV}{\mathcal{V}}

\newcommand{\matA}{{\boldsymbol{A}}}
\newcommand{\matB}{{\boldsymbol{B}}}
\newcommand{\matC}{{\boldsymbol{C}}}
\newcommand{\matD}{{\boldsymbol{D}}}
\newcommand{\matI}{{\boldsymbol{I}}}
\newcommand{\matJ}{{\boldsymbol{J}}}
\newcommand{\matP}{{\boldsymbol{P}}}
\newcommand{\matQ}{{\boldsymbol{Q}}}
\newcommand{\matR}{{\boldsymbol{R}}}
\newcommand{\matS}{{\boldsymbol{S}}}
\newcommand{\matU}{{\boldsymbol{U}}}
\newcommand{\matV}{{\boldsymbol{V}}}
\newcommand{\matX}{{\boldsymbol{X}}}
\newcommand{\matY}{{\boldsymbol{Y}}}
\newcommand{\matzero}{{\boldsymbol{0}}}
\newcommand{\matSigma}{\Upsigma}

\newcommand{\veca}{{\boldsymbol{a}}}
\newcommand{\vecb}{{\boldsymbol{b}}}
\newcommand{\vecu}{{\boldsymbol{u}}}
\newcommand{\vecv}{{\boldsymbol{v}}}
\newcommand{\vecw}{{\boldsymbol{w}}}
\newcommand{\vecx}{{\boldsymbol{x}}}
\newcommand{\vecy}{{\boldsymbol{y}}}
\newcommand{\vecz}{{\boldsymbol{z}}}
\newcommand{\veczero}{{\boldsymbol{0}}}

\newcommand{\rmatA}{{\boldsymbol{\mathsf{A}}}}
\newcommand{\rmatB}{{\boldsymbol{\mathsf{B}}}}
\newcommand{\rmatX}{{\boldsymbol{\mathsf{X}}}}

\newcommand{\rveca}{{\boldsymbol{\mathsf{a}}}}
\newcommand{\rvecb}{{\boldsymbol{\mathsf{b}}}}
\newcommand{\rvecx}{{\boldsymbol{\mathsf{x}}}}

\newcommand{\naturals}{\mathbb N}
\newcommand{\reals}{\mathbb R}
\newcommand{\diag}{diag}			
\newcommand{\rank}{rank}	
\newcommand{\tp}[1]{\ensuremath{#1^{\mathsf{T}}}} 
\newcommand{\tr}{tr}
\newcommand{\opP}{\operatorname{P}}
\newcommand{\opE}{\operatorname{E}}
\newcommand{\lebmeasure}{\lambda}
\newcommand{\ind}[1]{\chi_{#1}}	
					
\documentclass{article}
\usepackage[a4paper, total={6in, 8in}]{geometry}
\usepackage{cite}
\usepackage{mathtools}
\usepackage{enumitem}
\usepackage{mathrsfs}
\usepackage{amssymb}
\usepackage{amsthm}
\usepackage{amsmath}
\usepackage[capitalise]{cleveref}
\crefname{equation}{}{}
\crefrangeformat{section}{Sections #3#1#4--#5#2#6}
\usepackage[Symbol]{upgreek}

\newtheorem{theorem}{Theorem}[section]
\newtheorem{proposition}{Proposition}[section]
\newtheorem{definition}{Definition}[section]
\newtheorem{lemma}{Lemma}[section]
\newtheorem{example}{Example}[section]
\newtheorem{corollary}{Corollary}[section]

\numberwithin{equation}{section}

\begin{document}

\title{Completion of Matrices with Low Description Complexity}
\author{E. Riegler,  G. Koliander, D. Stotz, and H. B\"olcskei}
\date{}



\allowdisplaybreaks[3]



\maketitle

\begin{abstract}
We propose a theory for matrix completion that goes beyond the low-rank structure commonly considered in the literature and applies to general matrices of low description complexity. Specifically, complexity of the sets of matrices encompassed by the theory is measured in terms of Hausdorff and upper Minkowski dimensions.
Our goal is the characterization of the number of linear measurements, with an emphasis on  rank-$1$ measurements, needed for the existence of an algorithm that yields
reconstruction, either perfect, with probability 1, or with arbitrarily small probability of error, depending on the setup.
Concretely, we show that matrices taken from a set $\setU$ such that $\setU-\setU$ has Hausdorff dimension $s$ 
can be recovered from $k>s$ measurements, and random matrices supported on a set $\setU$ of Hausdorff dimension $s$ 
can be recovered with probability 1 from $k>s$ measurements. What is more, we establish the existence of 
recovery mappings that are robust against additive perturbations or noise in the measurements. Concretely, we show that there are 
 $\beta$-H\"older continuous mappings recovering
matrices taken from a set of upper Minkowski dimension $s$ from $k>2s/(1-\beta)$ measurements and, with arbitrarily small probability of error,
random matrices supported on a set of upper Minkowski dimension $s$ from $k>s/(1-\beta)$ measurements.
The numerous concrete examples we consider include low-rank matrices, sparse matrices, QR decompositions with sparse R-components, and matrices 
of fractal nature.
\end{abstract}



\section{Introduction} 

Matrix completion refers to the recovery of a low-rank matrix from a (small) subset of its entries or a (small) number of linear combinations thereof.
This problem arises in a wide range of applications including quantum state tomography, face recognition, recommender systems, and sensor localization (see, e.g., \cite{capl11,dogamo13} and references therein).  

The formal problem statement is as follows. Suppose we have $k$ linear measurements of the matrix $\matX\in\reals^{m\times n}$  with $\rank(\matX)\leq r$ in the form of
\begin{align}\label{eq:smodel}
\vecy
&=\tp{(\langle\matA_1,\matX\rangle\; \dots \;\langle\matA_k,\matX\rangle)},  
\end{align}
where $\langle \matA_i,\matX \rangle = \tr (\tp{\matA_i} \matX)$ is the standard trace inner product between matrices; 
the
$\matA_i\in \reals^{m\times n}$ are referred to as measurement matrices. 
The number of measurements, $k$, is  typically much smaller than the total number of entries, $mn$, of $\matX$. 
Depending on the  $\matA_i$, the measurements can simply be individual entries of $\matX$ or general linear combinations thereof. Can we recover $\matX$ from $\vecy$?

The vast literature on matrix completion (for a highly incomplete list see \cite{gr11,care09,cata10,refapa10,capl10,kemooh10,kemooh10a,capl11, re11,cazh14,cazh15,tabadr12,dogamo13,klopp2019structured}) provides thresholds on the number of measurements needed for successful recovery of the unknown low-rank matrix $\matX$, under various assumptions on the measurement matrices $\matA_i$ and the low-rank models generating $\matX$. 
For instance, 
in \cite{gr11} the  $\matA_i$ are chosen randomly from a fixed  orthonormal (w.r.t.\ the trace inner product) basis for $\reals^{n\times n}$ and it is shown that an unknown $n\times n$ matrix $\matX$ of rank no more than $r$ can be recovered with high probability if $k\geq\mathcal{O}(nr\nu\ln^2 n)$. 
Here, $\nu$ denotes the  coherence \cite[Def.~1]{gr11} between the unknown matrix $\matX$ and the orthonormal basis the $\matA_i$ are drawn from. 

The setting in \cite{care09} assumes random%
\footnote{We indicate random quantities by roman sans-serif letters such as  $\rmatA$.}
 measurement matrices $\rmatA_i$  with the position of the only nonzero entry, which is equal to one, chosen uniformly at random.  
It is shown that almost all (a.a.)  matrices $\matX$ (where a.a. is with respect to the random orthogonal model \cite[Def. 2.1]{care09}) of rank no more than $r$ can be recovered with high probability (with respect to the measurement matrices) provided that the number of measurements satisfies $k\geq Cn^{1.25}r\ln n$, where $C$ is a constant. 

In  \cite{capl11} it is shown that for random measurement matrices $\rmatA_i$  satisfying a certain concentration property, matrices $\matX$ of rank no more than $r$ can be recovered with high probability from $k\geq C(m+n)r$ measurements, where $C$ is a constant. 

The results on recovery thresholds reviewed so far as well as those in   \cite{cata10,refapa10, re11,cazh14,cazh15} all pertain to recovery through nuclear norm minimization. 

In \cite{elnepl12} measurement matrices $\rmatA_i$ containing i.i.d.\  entries drawn from an absolutely continuous (with respect to Lebesgue measure) distribution are considered.  
It is shown that rank minimization (which is NP-hard, in general) recovers $n\times n$ matrices $\matX$ of rank no more than $r$ with probability 1 if $k> (2n-r)r$. 
Furthermore, it is established that all matrices $\matX$ of rank no more than $n/2$ can be recovered, again through rank minimization and with probability 1, provided that  $k\geq 4nr-4r^2$.  

The recovery thresholds in  \cite{elnepl12},\cite{capl11} do not exhibit a $(\log n)$-factor, but assume significant  richness in the random measurement matrices $\rmatA_i$. 
Storing and applying the realizations of such measurement matrices is costly in terms of memory and computation time, respectively.  
To alleviate this problem, \cite{cazh15} considers rank-1 measurement matrices of the form $\rmatA_i=\rveca_i\tp{\rvecb_i}$, where the random vectors $\rveca_i \in \reals^m$ and $\rvecb_i \in \reals^n$ 
are independent with i.i.d.\  Gaussian or sub-Gaussian entries;  it is shown that nuclear norm minimization succeeds under the same recovery threshold  as in \cite{capl11}, namely $k\geq C(m+n)r$ for some constant $C$. 

The recovery of matrices that, along with the measurement matrices, belong to algebraic varieties was studied in  \cite{rong2021almost,xu18,waxu19}. As a byproduct, it is shown that
almost all rank-$r$ matrices in $\reals^{m\times n}$ can  be recovered from $k> (m+n-r)r$ measurements taken with measurement matrices of arbitrary rank.

Finally, the application of recent results on analog signal compression \cite{wuve10,striagbo15,albole19,bagusp19} to matrix completion yields recovery thresholds for a.a. measurement matrices $\matA_i$ and random matrices  $\rmatX$ that have low description complexity in the sense of \cite{albole19}. Specifically, the results in \cite{albole19} can be transferred to matrix completion by writing the trace inner product $\langle \matA_i,\matX \rangle$  as the standard inner product between the vectorized matrices $\matA_i$ and $\matX$ (obtained by stacking the columns). 
The definition of ``low description complexity'' as put forward in \cite{albole19} goes beyond the usual assumption of $\rmatX$ having low rank. 
It essentially says that the matrix takes value in some low-dimensional%
\footnote{The precise dimension measures used in the definition of low description complexity in  \cite{wuve10,striagbo15,albole19,bagusp19}  are, depending on the context, lower modified Minkowski dimension, upper Minkowski dimension, or Hausdorff dimension.} 
set $\setU\subseteq\reals^{m\times n}$ with probability 1. 
This set  $\setU$ can, for example, be the set of all matrices with rank no more than $r$, but much more general structures are possible.

\emph{Contributions.} 
The purpose of this paper is to  establish fundamental recovery thresholds (i.e., thresholds not restricted to a certain recovery scheme) for 
rank-1 measurement matrices $\matA_i=\veca_i\tp{\vecb_i}$ applied to data matrices $\matX$ taking value in low-dimensional  sets $\setU\subseteq\reals^{m\times n}$.
Rank-1 measurement matrices are practically relevant due to reduced storage requirements and lower computational complexity in the evaluation of the trace inner product  
$\langle \veca_i\tp{\vecb_i}, \matX\rangle= \tp{\veca_i}\matX\vecb_i$.
We consider both deterministic data matrices $\matX$ with associated recovery guarantees for all $\matX \in \setU$ and random $\rmatX$ accompanied by recovery guarantees either 
with probability 1 or with arbitrarily small probability of error. 
The recovery thresholds we obtain are in terms of the Hausdorff dimension of the support set $\setU$ of the data matrices. 
Furthermore, we establish bounds--in terms of the upper Minkowski dimension of $\setU$--on the number of measurements needed to guarantee 
H\"older continuous recovery, and hence robustness  against additive perturbations or noise. 
Hausdorff and upper Minkowski dimension are particularly easy to characterize for countably rectifiable and rectifiable sets \cite{fed69}, respectively.
These concepts comprise most practically relevant data structures such as low rank and sparsity in terms of the number of nonzero entries as well as
the Kronecker product, matrix product, or sum of any such matrices. 
As an example of sets that do not fall into the rich class of rectifiable sets, but our theory still applies to, we consider  sets of fractal nature. Specifically, we investigate 
 attractor sets of  recurrent iterated function systems \cite{baelha89}. 
 
Recovery thresholds for general (as opposed to rank-1) measurement matrices follow in a relatively straightforward manner through vectorization from the theory of lossless analog compression as developed in  \cite{albole19}. 
For the reader's convenience, we shall describe these extensions in brief wherever appropriate.
We finally note that a preliminary version of part of the work reported in the present paper, specifically weaker results for more restrictive sets $\setU$ of bounded Minkowski dimension and without  statements on H\"older-continuous recovery, was presented in \cite{ristbo15} by a subset of the authors.

\emph{Organization of the paper.} 
In \cref{sec:mainres}, we present recovery thresholds for matrices $\matX$ taking value in a general set $\setU\subseteq\reals^{m\times n}$.
These results are formulated in terms of Hausdorff and upper Minkowski dimension of $\setU$.
 In \cref{sec:1}, we introduce the concept of rectifiable and countably rectifiable sets from geometric measure theory \cite{fed69} and we characterize the upper Minkowski and Hausdorff dimensions of such sets.
Furthermore, it is shown that many practically relevant sets of structured matrices are (countably) rectifiable. The particularization of our general recovery thresholds to the rectifiable case concludes this section. In \cref{sec:RIFS}, we particularize our general recovery thresholds to  
 attractor sets of  recurrent iterated function systems.  
The proofs of our main results, \cref{th1,th2} and \cref{prp:ns,prp:nsreg} are contained in \crefrange{sec:proofth1}{nsregproof}.

\emph{Notation.} 
Capitalized boldface letters $\matA,\matB,\dots$  designate deterministic matrices and lowercase boldface letters $\veca,\vecb,\dots$ stand for deterministic vectors. 
We use roman sans-serif letters for random quantities (e.g., $\rvecx$ for a random vector and $\rmatA$ for a random matrix). 
Random quantities are assumed to be defined on the Borel $\sigma$-algebra of the underlying space. $\opP[\rmatX\in\setU]$ denotes the probability of $\rmatX$ being in the Borel set $\setU$.
We write $\lebmeasure$ for  Lebesgue measure.  
The superscript  $\tp{}$ stands for transposition. 
The  ordered singular values of a matrix $\matA$ are denoted by $\sigma_1(\matA)\geq\dots\geq\sigma_n(\matA)$.
We write $\matA\otimes\matB$ for  the Kronecker product of the matrices $\matA$ and $\matB$, denote the trace of $\matA$ by  
$\tr(\matA)$, 
and let   
$\langle\matA,\matB\rangle=\tr(\tp{\matA}\matB)$ be the trace inner product of $\matA$ and $\matB$.  
Further, 
$\lVert\matA\rVert_2=\sqrt{\langle\matA,\matA\rangle}$ and   $\lVert\matA\rVert_0$ refers to the number of  nonzero entries of $\matA$.  
For the Euclidean space $(\reals^k,\lVert\,\cdot\,\rVert_2)$, we denote the open ball of radius $s$ centered at $\vecu\in \reals^k$ by $\setB_k(\vecu,s)$, $V(k,s)$ stands for its volume.  
Similarly, forthe Euclidean space $(\reals^{m\times n},\lVert\,\cdot\,\rVert_2)$, the open ball of radius $s$  centered at $\matU\in \reals^{m\times n}$ is  $\setB_{m\times n}(\matU,s)$. 
We set $\setM_r^{m\times n}=\{\matX\in\reals^{m\times n}:\rank(\matX)\leq r\}$ 
and let $\setA_s^{m\times n}=\{\matX\in\reals^{m\times n}:\lVert\matX\rVert_0\leq s\}$.  
The closure of the set $\setU$ is denoted by $\overline{\setU}$. 
The Cartesian product of the sets $\setA$ and $\setB$ in Euclidean space
is written as $\setA\times \setB$ and their Minkowski difference is   designated by $\setA-\setB$. 
The indicator function of a set $\setU$ is designated by $\chi_{\setU}$.
For a bounded set $\setU\subseteq\reals^{m\times n}$ and $\delta>0$, we denote the covering number of $\setU$ by  
\begin{align}
N_{\delta}(\setU)=\min\Bigg\{N\in\naturals: \exists \matY_1,\dots,\matY_N\in\setU \textrm{ s.t. } \setU\subseteq\bigcup_{i=1}^N\setB_{m\times n}(\matY_i,\delta)\Bigg\}. 
\end{align}
For $\setU,\setV\subseteq\reals^{m\times n}$, we set $\setU-\setV=\{\matU-\matV:\matU\in\setU,\matV\in\setV\}$. 
We let $\dim_\mathrm{H}(\cdot)$ denote the Hausdorff dimension  \cite[Equation (3.10)]{fa14}, defined by 
\begin{equation}
  \dim_\mathrm{H}(\setU) = \inf\{s \geq 0: \mathscr{H}^s(\setU) = 0\},
\end{equation} 
where $\mathscr{H}^s(\setU) = \lim_{\delta \to 0} \inf \big\{\sum_{i=1}^\infty \operatorname{diam}^s(\setU_i): \{\setU_i\} \text{ is a $\delta$-cover of $\setU$} \big\}$ is the $s$-di\-men\-sion\-al Hausdorff measure of $\setU$ \cite[Equation (3.2)]{fa14}.

Furthermore, $\overline{\dim}_\mathrm{B}(\cdot)$  and $\underline{\dim}_\mathrm{B}(\cdot)$ refer to the upper and lower Minkowski dimension \cite[Definition 2.1]{fa14} defined as 
\begin{equation}\label{eq:dimbupper}
  \overline{\dim}_\mathrm{B}(\setU) = \operatorname{lim\,sup}_{\delta \to 0} \frac{\log N_{\delta}(\setU)}{\log (1/\delta)}
\end{equation} 
and 
\begin{equation}\label{eq:dimblower}
  \underline{\dim}_\mathrm{B}(\setU) = \operatorname{lim\,inf}_{\delta \to 0} \frac{\log N_{\delta}(\setU)}{\log (1/\delta)},
\end{equation} 
respectively.
Finally, we  note that  \cite[Proposition 3.4]{fa14}
\begin{align}\label{eq:HMB}
\dim_\mathrm{H}(\setU) \leq  \underline{\dim}_\mathrm{B}(\setU)  \leq  \overline{\dim}_\mathrm{B}(\setU) 
\end{align}
for all nonempty  subsets in Euclidean spaces. 

\section{Main Results}
\label{sec:mainres}

Our first main result provides a threshold for recovery of matrices $\matX$ from rank-1 measurements  in a very general setting. Specifically, the matrices $\matX$ are assumed to 
take value in some set $\setU$ and the recovery threshold is in terms of either $\dim_\mathrm{H}(\setU)$ or
$\dim_\mathrm{H}(\setU-\setU)$.

\begin{theorem}\label{th1}
For every nonempty set $\setU\subseteq\reals^{m\times n}$,  the following holds:
\begin{enumerate}[label=\roman*)]
\item\label{th1item1}
The mapping 
\begin{align}
\setU&\to \reals^{k}\label{eq:linearity1}\\
\matX&\mapsto \tp{(\tp{\veca_1}\matX\vecb_1\ \dots\ \tp{\veca_k}\matX\vecb_k)}\label{eq:linearity2}
\end{align}
is one-to-one for 
Lebesgue a.a. $((\veca_1 \dots \veca_k), (\vecb_1 \dots \vecb_k))  \in\reals^{m\times k}\times \reals^{n\times k}$
provided that $\dim_\mathrm{H}(\setU-\setU)<k$.
\item \label{th1item2}
Suppose that $\setU$  is Borel and 
consider an $m\times n$ random matrix  $\rmatX$ 
 satisfying $\opP[\rmatX\in\setU]=1$.     
Then, for Lebesgue a.a. $((\veca_1 \dots \veca_k), (\vecb_1 \dots \vecb_k))  \in\reals^{m\times k}\times \reals^{n\times k}$,  
there exists  
a  Borel-measurable mapping $g\colon \reals^k\to\reals^{m\times n}$ satisfying  
\begin{align}\label{eq:decoder}
\opP\Big[g\Big( \tp{(\tp{\veca_1}\rmatX\vecb_1\ \dots\ \tp{\veca_k}\rmatX\vecb_k)}\Big)\neq\rmatX\Big]=0  
\end{align}
 provided that $\dim_\mathrm{H}(\setU)<k$. 
\end{enumerate}
\end{theorem}
\begin{proof}
See  \cref{sec:proofth1}.
\end{proof}
The first part of the theorem states that in the deterministic case, $k>\dim_\mathrm{H}(\setU-\setU)$ rank-1 measurements suffice for unique recovery of $\matX\in\setU$ (except for measurement vectors $\veca_i,\vecb_i$ supported on a Lebesgue null-set). In the probabilistic setting of the second part, $k>\dim_\mathrm{H}(\setU)$ measurements suffice for the existence of a Borel-measurable recovery mapping achieving zero error. While these are the most general versions of our recovery results, it can be difficult to evaluate $\dim_\mathrm{H}(\setU)$ and $\dim_\mathrm{H}(\setU -\setU)$ for sets $\setU$ with interesting
structural properties.
In \cref{sec:1}, we shall see, however, that these dimensions are easily computed for rectifiable sets, which, in turn, encompass many structures of practical relevance such as low rank or sparsity.

A  version of  \cref{th1} in terms of lower Minkowski dimension instead of Hausdorff dimension 
was found by a subset of the authors of the present paper in \cite[Theorem~2]{ristbo15} and was subsequently extended by Li et al.\ to the complex-valued case in \cite[Theorem 8]{lilebr17}. 
As lower Minkowski dimension is always greater than or equal to Hausdorff dimension (see \cref{eq:HMB}), \cref{th1} strengthens \cite[Theorem~2]{ristbo15}.
In addition, lower Minkowski dimension is defined for bounded sets $\setU$ only, a restriction not shared by Hausdorff dimension.   

A vectorization argument, concretely, stacking the columns of the data matrix $\matX$ and the measurement matrices $\matA_i$, shows that \cite[Theorem~3.7]{bagusp19}
implies \cref{th1item1} in \cref{th1} and \cite[Corollary~3.4]{bagusp19} implies \cref{th1item2}, in both cases, however, for the rank-1 measurement matrices $\veca_i\tp{\vecb_i}$ replaced by generic measurement matrices $\matA_i\in \reals^{m\times n}$.
As the set of rank-1 matrices is a null-set when viewed as a subset of $\reals^{m\times n}$, these results do not imply  our \cref{th1}.
Furthermore, the technical challenges 
 in establishing \cref{th1} are quite different from those encountered in \cite{bagusp19}, which, in turn, builds on \cite{albole19}.
In particular, here we need a stronger concentration of measure inequality (see  \cref{lem:com1} in  \cref{prp:nsproof}).  

The proof of  \cref{th1} detailed in \cref{sec:proofth1} is based on  the following result, which is similar in spirit to the null-space property in compressed sensing theory \cite[Theorem~2.13]{fora13}:
\begin{proposition}\label{prp:ns} 
Consider a nonempty set $\setU\subseteq\reals^{m\times n}$ with $\dim_\mathrm{H}(\setU)<k$.
Then, 
\begin{align}\label{eq:proptoshow}
\{\matX\in\setU\!\setminus\!\{\matzero\}: \tp{(\tp{\veca_1}\matX\vecb_1\ \dots\ \tp{\veca_k}\matX\vecb_k)}=\veczero\}=\emptyset,
\end{align}
for Lebesgue a.a.  $((\veca_1 \dots \veca_k), (\vecb_1 \dots \vecb_k))  \in\reals^{m\times k}\times\reals^{n\times k}$.  
\end{proposition}
\begin{proof}
See \cref{prp:nsproof}.
\end{proof}

Our second main result establishes  thresholds for H\"older-continuous recovery,  that is, recovery which exhibits robustness against additive perturbations or noise. 
Here, we have to impose the stricter technical condition of bounded upper Minkowski dimension $\overline{\dim}_\mathrm{B}(\setU)$ and, in turn, can only consider bounded sets $\setU$.  
\begin{theorem}\label{th2}
For every nonempty and bounded set $\setU\subseteq\reals^{m\times n}$ and $\beta\in(0,1)$,  the following holds:
\begin{enumerate}[label=\roman*)]
\item\label{th2item1}
Suppose that 
\begin{align}\label{eq:assdimB2a}
\frac{\overline{\dim}_\mathrm{B}(\setU-\setU)}{1-\beta}<k.  
\end{align}
Then, for Lebesgue a.a. $((\veca_1 \dots \veca_k), (\vecb_1 \dots \vecb_k))  \in\reals^{m\times k}\times\reals^{n\times k}$,  there exists a $\beta$-H\"older continuous mapping $g\colon \reals^{k}\to\reals^{m\times n}$ satisfying 
\begin{align}
g\Big(\tp{(\tp{\veca_1}\matX\vecb_1\ \dots\ \tp{\veca_k}\matX\vecb_k)}\Big)=\matX, \quad \text{for all $\matX\in\setU$.} 
\end{align}
\item \label{th2item2}
Suppose that $\setU$ is Borel with 
\begin{align}\label{eq:assdimB2}
\frac{\overline{\dim}_\mathrm{B}(\setU)}{1-\beta}<k. 
\end{align}
Fix  $\varepsilon>0$ arbitrarily and consider an $m\times n$ random matrix  $\rmatX$ with  $\opP[\rmatX\in\setU]=1$. 
Then,  for Lebesgue a.a.  $((\veca_1 \dots \veca_k), (\vecb_1 \dots \vecb_k))  \in\reals^{m\times k}\times\reals^{n\times k}$,  
there exists a $\beta$-H\"older continuous mapping $g\colon \reals^{k}\to\reals^{m\times n}$ satisfying 
\begin{align}
\opP\Big[g\Big(\tp{(\tp{\veca_1}\rmatX\vecb_1\ \dots\ \tp{\veca_k}\rmatX\vecb_k)}\Big)\neq \rmatX\Big]\leq\varepsilon. 
\end{align}
\end{enumerate}
\end{theorem}
\begin{proof}
See  \cref{sec:proofth2}.
\end{proof}
Again, the first part of the theorem concerns deterministic data matrices $\matX$ for which  $k>\overline{\dim}_\mathrm{B}(\setU-\setU)/(1-\beta)$  rank-1 measurements (except for measurement vectors supported on a Lebesgue null-set) guarantee  $\beta$-H\"older continuous recovery. The higher the desired H\"older exponent $\beta$, the larger the number of measurements has to be. 
In the probabilistic setting of the second part of the theorem,  $k>\overline{\dim}_\mathrm{B}(\setU)/(1-\beta)$  measurements suffice for $\beta$-H\"older continuous recovery.  
 We hasten to add that recovery 
is  only with probability $1-\varepsilon$, where, however, $\varepsilon$ can be arbitrarily small. Also note that the number of measurements, $k$, is independent of $\varepsilon$. 
We shall evaluate $\overline{\dim}_\mathrm{B}(\setU)$ for several rectifiable sets with interesting structural properties in \cref{sec:1} and for attractor sets of  recurrent iterated function systems in \cref{sec:RIFS}.

A version of \cref{th2} for the rank-1 measurements replaced by measurements taken with general matrices  
can be obtained from results available in the literature. Specifically, the equivalent of \Cref{th2item1} in \cref{th2} follows from  \cite[Theorem 4.3]{ro11}, that of \Cref{th2item2} is obtained from 
\cite[Theorem 2]{striagbo15},
in both cases by vectorization.

The proof of  \cref{th2} is again based on a variant of the null-space property as used in compressed sensing theory, concretely on the following result: 
\begin{proposition}\label{prp:nsreg} 
Consider a  nonempty and bounded set $\setU\subseteq\reals^{m\times n}$, 
and suppose that there exists a $\beta\in (0,1)$ such that 
\begin{align}\label{eq:assdimB}
\frac{\overline{\dim}_\mathrm{B}(\setU)}{1-\beta}<k.
\end{align}
Then, 
\begin{align}\label{eq:proptoshowreg}
\inf \Bigg\{\frac{\lVert\tp{(\tp{\veca_1}\matX\vecb_1\ \dots\ \tp{\veca_k}\matX\vecb_k)}\rVert_2}{\lVert\matX\rVert_2^{1/\beta}}: \matX\in\setU\!\setminus\!\{\matzero\} \Bigg\}>0,
\end{align}
for Lebesgue a.a.  $((\veca_1 \dots \veca_k), (\vecb_1 \dots \vecb_k))  \in\reals^{m\times k}\times\reals^{n\times k}$.   
\end{proposition}
\begin{proof}
See \cref{prp:nsproof}.
\end{proof}

Regarding converse statements, i.e., the question of whether too few rank-1 measurements of a given random matrix $\rmatX$ necessarily render unique reconstruction impossible, 
we note that \cite[Corollary IV.2]{albole19} allows a partial answer. Specifically, the 
simple characterization of the support set $\setU$ of $\rmatX$ through its dimension $\overline{\dim}_\mathrm{B}(\setU)$
does not enable a general impossibility result.
If one assumes, however, that the vectorized version of $\rmatX$ is $k$-analytic according to \cite[Definition~IV.2]{albole19}, then we can conclude that fewer than $k$ measurements necessarily lead to
reconstruction of $\rmatX$ being impossible, with probability 1.
This statement holds for arbitrary measurement matrices, so in particular also for rank-1 matrices.

\section{Rectifiable Sets}\label{sec:1}
To illustrate the practical applicability of the general recovery thresholds obtained in \cref{th1,th2} and expressed in terms of Hausdorff and  upper  Minkowski dimension, we 
first introduce the concept of rectifiable sets, a central element of 
geometric measure theory \cite{fed69}. The relevance of rectifiability derives itself from the fact that a broad class of  structured data matrix support sets  we are interested in  turns  out to be rectifiable. In addition,  
Hausdorff and upper Minkowski dimensions of rectifiable sets have been characterized in significant detail in the literature.

We start with the formal definition of rectifiable sets. 
\begin{definition}\cite[Definition 3.2.14]{fed69}\label{Dfn:recset}
For  $s\in\naturals$,  the set $\setU\subseteq\reals^{m\times n}$ is 
\begin{enumerate}[label=\roman*)]
\item $s$-rectifiable if there exist a nonempty and compact set $\setA\subseteq\reals^s$ and a Lip\-schitz  mapping $\varphi\colon\setA\to \reals^{m\times n}$  such that $\setU=\varphi(\setA)$; \label{recd1}
\item countably $s$-rectifiable if it is the countable union of $s$-rectifiable sets;\label{recd2}
\item countably $\mathscr{H}^s$-rectifiable if it is $\mathscr{H}^s$-measurable and there exists a countably $s$-rectifiable set $\setV\subseteq\reals^{m\times n}$  such that $\mathscr{H}^s(\setU\setminus \setV)=0$. \label{recd3}
\end{enumerate}
\end{definition} 
We have the following obvious chain of implications: 

\begin{center} 
$s$-rectifiable $\Rightarrow$  countably $s$-rectifiable $\Rightarrow$ countably $\mathscr{H}^{s}$-rectifiable.
\end{center} 

Countably $\mathscr{H}^{s}$-rectifiable sets  thus constitute the most general class. 

We proceed to state preparatory results, which will be used later to establish that many practically relevant sets of structured matrices are (countably) $s$-rectifiable and to quantify the associated rectifiability parameter $s$.

\begin{lemma} \label{LemmaExarec2}(Properties of $s$-rectifiable sets)
\begin{enumerate}[label=\roman*)]
\item \label{exarec}
If $\setU\subseteq\reals^{m\times n}$ is $s$-rectifiable, then it is $t$-rectifiable for all $t\in\naturals$  with $t>s$. 
\item \label{exasum2}
For  $\setU_i \subseteq\reals^{m\times n}$ $s_i$-rectifiable with $s_i\leq s$, $i=1,\dots,N$,  the set
\begin{align}
\setU=\bigcup_{i=1}^N\setU_i 
\end{align}
is $s$-rectifiable.
In particular, the finite union of  $s$-rectifiable sets is $s$-rec\-ti\-fi\-able.
\item \label{exaProdrec2}
If $\setU\subseteq\reals^{m_1\times n_1}$ is $s$-rectifiable and  $\setV\subseteq\reals^{m_2\times n_2}$ is $t$-rectifiable, then 
$\setU\times\setV$ 
is  $(s+t)$-rectifiable. 
\item \label{exaC12}
Every compact subset of an $s$-dimensional $C^1$-submanifold \cite[Definition 5.3.1]{krpa08} of $\reals^{m\times n}$ is $s$-rectifiable.
\end{enumerate}
\end{lemma}
\begin{proof}See \cref{ProofLemmaExarec2}.\end{proof}

\begin{lemma} \label{LemmaExarec}(Properties of countably $s$-rectifiable sets)
\begin{enumerate}[label=\roman*)]
\item \label{exaProdrec}
If $\setU\subseteq\reals^{m_1\times n_1}$ is countably  $s$-rectifiable and  $\setV\subseteq\reals^{m_2\times n_2}$ is countably $t$-rectifiable, then 
$\setU\times\setV$ 
is countably $(s+t)$-rectifiable.  
\item \label{exasum}
For  $\setU_i \subseteq\reals^{m\times n}$ countably $s_i$-rectifiable with $s_i\leq s$, $i\in\naturals$,  the set
\begin{align}
\setU=\bigcup_{i\in\naturals}\setU_i 
\end{align}
is  countably $s$-rectifiable.  
\item \label{exaC1}
Every $s$-dimensional $C^1$-submanifold \cite[Definition 5.3.1]{krpa08} of $\reals^{m\times n}$ is countably  $s$-rectifiable. 
In particular, every $s$-dimensional affine subspace of $\reals^{m\times n}$ is countably $s$-rectifiable. 
\end{enumerate}
\end{lemma}
\begin{proof}Follows from  \cite[Lemma III.1]{albole19}.\end{proof}

To establish  the rectifiability of structured matrices obtained as products or sums of structured matrices, we need to understand the impact of continuous mappings on rectifiability. 
Specifically, we shall need the following result from  \cite[Lemma III.~3]{albole19} for locally-Lipschitz mappings, i.e., functions that are Lipschitz continuous on all compact subsets:

\begin{lemma}\label{lem:recsetlip}
Let $\setU\subseteq\reals^{m_1\times n_1}$ and  let $f\colon\reals^{m_1\times n_1}\to\reals^{m_2\times n_2}$ be a locally-Lipschitz mapping.  
\begin{enumerate}[label=\roman*)]
\item 
If $\setU$ is $s$-rectifiable, then $f(\setU)$ is $s$-rectifiable. \label{recf1}
\item If $\setU$ is 
countably $s$-rectifiable, then $f(\setU)$ is countably $s$-rectifiable. \label{recf2}
\end{enumerate}
\end{lemma}

We will mainly use the following generalization of \cref{lem:recsetlip}:
\begin{lemma}\label{lem:forexa}
Consider a locally-Lipschitz mapping  $f\colon \bigtimes_{i=1}^N \reals^{m_i\times n_i}\to \reals^{m\times n}$,  and suppose that $\setU_i\subseteq\reals^{m_i\times n_i}$, for $i=1,\dots,N$. 
\begin{enumerate}[label=\roman*)]
\item If $\setU_i$ is $s_i$-rectifiable, for $i=1,\dots,N$,   then $f(\setU_1\times\dots\times\setU_N)$ is $s$-rectifiable with $s=\sum_{i=1}^Ns_i$. \label{itemrec}
\item 
If $\setU_i$ is countably $s_i$-rectifiable, for $i=1,\dots,N$,  then $f(\setU_1\times\dots\times\setU_N)$ is countably $s$-rectifiable with  $s=\sum_{i=1}^Ns_i$. \label{itemcrec}
\end{enumerate}
\end{lemma}
\begin{proof}
\Cref{itemrec} follows from \cref{exaProdrec2} of  \cref{LemmaExarec2} and \cref{recf1} of \cref{lem:recsetlip}, and 
\cref{itemcrec} follows from \cref{exaProdrec} of  \cref{LemmaExarec} and \cref{recf2} of \cref{lem:recsetlip}. 
\end{proof}
  
Before we can particularize our results in \cref{th1,th2}, it remains to characterize the Hausdorff dimension and the upper Minkowski dimension of 
 rectifiable sets in terms of their rectifiability parameters.   

 \begin{lemma}\label{lem:recdim}
Let   $\setU \subseteq\reals^{m\times n}$ be nonempty. Then, the following properties hold: 
\begin{enumerate}[label=\roman*)]
\item If  $\setU$ is  countably $\mathscr{H}^s$-rectifiable,  then 
\begin{align}
\dim_{\mathrm{H}}(\setU)\leq s.
\end{align}\label{item1rec}
\item 
If  $\setU\subseteq \setV$  with $\setV\subseteq\reals^{m\times n}$ $s$-rectifiable,  then 
  \begin{align}
  \overline{\dim}_\mathrm{B}(\setU)\leq s.
  \end{align}
  \label{item2rec}
\end{enumerate}
\end{lemma}
\begin{proof}
We first prove \cref{item1rec}. Since $\setU$ is countably $\mathscr{H}^s$-rectifiable,  
by \cref{Dfn:recset}, there exists a  countably $s$-rectifiable set $\setV\subseteq\reals^{m\times n}$  with 
$\mathscr{H}^s(\setU\setminus \setV)=0$.  
By \cite[Lemma III.2]{albole19}, the upper modified Minkowski dimension of a countably $s$-rectifiable set $\setV$ is upper-bounded by $s$. 
Combined with \cite[Equation (3.27)]{fa14}, which states that the Hausdorff dimension of $\setV$ is upper-bounded by the upper modified Minkowski dimension, this yields 
\begin{align}\label{eq:dimHV}
\dim_\mathrm{H}(\setV)\leq s. 
\end{align}
Since  $\mathscr{H}^s(\setU\setminus \setV)=0$, the definition of Hausdorff dimension implies  
\begin{align}\label{eq:bounds}
\dim_{\mathrm{H}}(\setU\setminus \setV)\leq s
\end{align} so that 
\begin{align}
\dim_{\mathrm{H}}(\setU) 
&=\max\{ \dim_{\mathrm{H}}(\setU\cap\setV), \dim_{\mathrm{H}}(\setU\setminus\setV) \}\label{eq:stabilityI}\\
&\leq \max\{ \dim_{\mathrm{H}}(\setV), \dim_{\mathrm{H}}(\setU\setminus\setV) \}\label{eq:monotonicityI}\\
&\leq s,\label{eq:usecountrect} 
\end{align}
where  \cref{eq:stabilityI} follows from countable  stability of Hausdorff dimension \cite[Section 3.2]{fa14}, in  \cref{eq:monotonicityI} we used  monotonicity of Hausdorff dimension \cite[Section 3.2]{fa14}, and \cref{eq:usecountrect} is by \cref{eq:dimHV} and \cref{eq:bounds}.

To establish \cref{item2rec}, we note that, by \cref{Dfn:recset}, a nonempty $s$-rectifiable set $\setV$ can be written as  $\setV=\varphi(\setA)$ for a Lip\-schitz  mapping $\varphi\colon\setA\to \reals^{m\times n}$  and a  nonempty compact set $\setA\subseteq \reals^s$. 
We thus have 
\begin{align}
\overline{\dim}_\mathrm{B}(\setU)
&\leq \overline{\dim}_\mathrm{B}(\setV)\label{eq:chain1}\\
&\leq  \overline{\dim}_\mathrm{B}(\setA)\label{eq:chain2}\\
&\leq s,\label{eq:chain3}
\end{align}
where \cref{eq:chain1} and \cref{eq:chain3} follow from the monotonicity of upper Minkowski dimension \cite[Section 2.2]{fa14}   upon noting that the compact set $\setA$ is a subset of an open ball in $\reals^s$
of sufficiently large radius, which has upper Minkowski dimension $s$, and in \cref{eq:chain2} we applied 
\cite[Proposition 2.5, Item (a)]{fa14}. 
\end{proof}

The following result will be useful in particularizing our deterministic recovery thresholds in \Cref{th1item1} of  \Cref{th1} and \Cref{th2item1} of  \Cref{th2}
to rectifiable sets.

\begin{lemma}\label{lem:rect2U}
Let   $\setU\subseteq\reals^{m\times n}$ be nonempty. Then,   
 \begin{align}\label{eq:nottight0}
\dim_{H}(\setU-\setU)\leq \overline{\dim}_{\mathrm B}(\setU-\setU)\leq 2\, \overline{\dim}_{\mathrm B}(\setU).
\end{align}
If, in addition, $\setU$ is (countably) $s$-rectifiable, then $\setU-\setU$ is (countably) $2s$-rectifiable with 
\begin{align}\label{eq:nottight1}
\dim_{\mathrm H}(\setU-\setU)\leq 2s.
\end{align}
\end{lemma}
\begin{proof}
The first inequality  in \cref{eq:nottight0} is by  \cref{eq:HMB} and 
the second inequality  in \cref{eq:nottight0} follows from \cite[Proposition 2.5, Item (a)]{fa14} with $f(\matA_1, \matA_2)=\matA_1 - \matA_2$ and the product formula \cite[Equation (7.9)]{fa14}. 
The set $\setU-\setU$  is (countably) $2s$-rectifiable  owing to \cref{itemrec}  (\cref{itemcrec}) in \cref{lem:forexa}  with $f(\matA_1, \matA_2)=\matA_1 - \matA_2$. Together with \cref{item1rec} in \Cref{lem:recdim} this yields \cref{eq:nottight1}. 
\end{proof}

We are now in a position to particularize the results in \cref{th1,th2} to rectifiable sets.

\begin{theorem}\label{threct}(Recovery for rectifiable sets)
\begin{enumerate}[label=\roman*)]
\item\label{en:rect1}
Let $\setU\subseteq \reals^{m\times n}$  be nonempty with $\setU-\setU$ countably $\mathscr{H}^s$-rectifiable. 
 Then, for $k > s$ and 
Lebesgue a.a. 
$((\veca_1 \dots \veca_k), (\vecb_1 \dots \vecb_k))$ 
$\in \reals^{m\times k}\times \reals^{n\times k}$, 
every  $\matX\in \setU$  can be recovered uniquely from the rank-1 measurements
\begin{align}
\tp{(\tp{\veca_1}\matX\vecb_1\ \dots\ \tp{\veca_k}\matX\vecb_k)}.
\end{align}
\item\label{en:rect2}
Let $\setU,\setV\subseteq \reals^{m\times n}$ be nonempty with  $\setV$  $s$-rectifiable and $\setU-\setU \subseteq \setV$.  
Fix  $\beta \in (0,1-s/k)$  with $k>s$. Then, 
  for   
Lebesgue a.a. 
$((\veca_1 \dots \veca_k), (\vecb_1 \dots \vecb_k))$ 
$\in \reals^{m\times k}\times \reals^{n\times k}$, 
every  $\matX\in \setU$   can be recovered uniquely from the rank-1 measurements
\begin{align}
\tp{(\tp{\veca_1}\matX\vecb_1\ \dots\ \tp{\veca_k}\matX\vecb_k)}
\end{align}
by a $\beta$-H\"older continuous mapping $g$. 
\item\label{en:rect3} 
Let $\setU\subseteq \reals^{m\times n}$ be nonempty, Borel, and countably $\mathscr{H}^s$-rectifiable.  
Suppose that the random matrix  $\rmatX$ satisfies  $\opP[\rmatX\in\setU]=1$.
Then, for Lebesgue a.a.  
$((\veca_1 \dots \veca_k), (\vecb_1 \dots \vecb_k))  \in\reals^{m\times k}\times \reals^{n\times k}$,  
there exists 
a Borel-measurable mapping $g\colon\reals^k\to\reals^{m\times n}$, satisfying  
\begin{align}
\opP\Big[g\Big(\tp{(\tp{\veca_1}\rmatX\vecb_1\ \dots\ \tp{\veca_k}\rmatX\vecb_k)}\Big)\neq\rmatX\Big]=0
\end{align}
 provided that  $k > s$. 
\item\label{en:rect4}
Let $\setU,\setV\subseteq \reals^{m\times n}$ be nonempty with $\setU$ Borel, $\setV$  $s$-rectifiable, and $\setU\subseteq \setV$. Suppose that the random matrix 
 $\rmatX$ satisfies  $\opP[\rmatX\in\setU]=1$.  Fix $\varepsilon>0$ and let  $\beta \in (0,1-s/k)$  with $k > s$. Then, for
Lebesgue a.a. 
$((\veca_1 \dots \veca_k), (\vecb_1 \dots \vecb_k))  \in\reals^{m\times k}\times \reals^{n\times k}$,  
there exists 
a $\beta$-H\"older continuous mapping $g\colon \reals^k\to\reals^{m\times n}$ satisfying  
\begin{align}
\opP\Big[g\Big((\tp{{\tp{\veca_1}}\rmatX\vecb_1\ \dots\ \tp{\veca_k}\rmatX\vecb_k)}\Big)\neq\rmatX\Big] \leq \varepsilon.
\end{align}
\end{enumerate}
\end{theorem}
\begin{proof}
The proof is a straightforward combination of results already established:   
\begin{itemize}
\item
\Cref{en:rect1} is by   
\cref{item1rec} in \Cref{lem:recdim}  combined with  \cref{th1item1} in \cref{th1}.  
\item 
\Cref{en:rect2} is by   \cref{item2rec} in \Cref{lem:recdim}  combined with  \cref{th2item1} in \cref{th2}.   
 \item
\Cref{en:rect3} is by    
\cref{item1rec} in \Cref{lem:recdim}  combined with \cref{th1item2} in \cref{th1}.  
\item
\Cref{en:rect4} is by \cref{item2rec} in \Cref{lem:recdim} combined with  \cref{th2item2} in \cref{th2}.$\ $  
\end{itemize}
\end{proof}

We now apply \Cref{threct} to various interesting  structured sets and start with sparse matrices.   

\begin{example}\label{exa:sparse}
Let $\setA_s^{m\times n}$ be the set of $s$-sparse matrices in $\reals^{m\times n}$, i.e., the set of   matrices with at most $s$ nonzero entries.
Further, let $\setA_{\setI}^{m\times n}$ denote the set of matrices that have their nonzero entries indexed by  $\setI \subseteq \{1, \dots, m\} \times \{1, \dots, n\}$. 
Obviously, $\setA_{\setI}^{m\times n}$ is a linear subspace of $\reals^{m\times n}$ of dimension  $\lvert \setI\rvert$. 
By \cref{exaC1} of \cref{LemmaExarec}, the set $\setA_{\setI}^{m\times n}$ is hence countably $\lvert \setI\rvert$-rectifiable.
As $\setA_s^{m\times n} = \bigcup_{\setI: \lvert \setI\rvert=s}\setA_{\setI}^{m\times n}$, it follows from \cref{exasum} in \cref{LemmaExarec} that $\setA_s^{m\times n}$ is countably $s$-rectifiable. Also note that $\setA_s^{m\times n}-\setA_s^{m\times n}=\setA_{2s}^{m\times n}$ is countably $2s$-rectifiable. 

Similarly, for every bounded subset $\setU \subseteq \setA^{m\times n}_s$, $\overline{\setU}$ is $s$-rectifiable.
This follows by first noting that  $\overline{\setU}$ is compact in $\reals^{m\times n}$  
and, therefore, for given $\setI$, $\setA_{\setI}^{m\times n} \cap \, \overline{\setU}$ is a compact subset of the linear subspace $\setA_{\setI}^{m\times n}$. Hence, by  \cref{exaC12,exasum2} of \cref{LemmaExarec2},  $\overline{\setU} =\bigcup_{\setI: \lvert \setI\rvert=s}(\setA_{\setI}^{m\times n} \cap \overline{\setU})$  is $s$-rectifiable.

We can therefore  apply the corresponding items of 
 \Cref{threct}  to obtain recovery thresholds for Lebesgue a.a. 
$((\veca_1 \dots \veca_k), (\vecb_1 \dots \vecb_k)) \in\reals^{m\times k}\times\reals^{n\times k}$ for the following sets: 
\begin{enumerate}[label=\roman*)]
\item If $\setU\subseteq  \setA_s^{m\times n}$ is nonempty, then every $\matX\in\setU$  can be recovered uniquely from $k>2s$ measurements since 
$\setU-\setU\subseteq \setA_{2s}^{m\times n}$ is countably $\mathscr{H}^{2s}$-rectifiable. 
  \item    If  $\setU\subseteq \setA_s^{m\times n}$ is nonempty and bounded, then every  $\matX\in\setU$  can be recovered uniquely from $k>2s$ measurements by a $\beta$-H\"older continuous mapping with $\beta\in(0,1-2s/k)$  since  $\setV=\overline{\setU-\setU}\subseteq \setA_{2s}^{m\times n}$ 
   is $(2s)$-rectifiable. 
\item If $\setU\subseteq \reals^{m\times n}$  is nonempty, Borel,  and satisfies  $\mathscr{H}^s(\setU\setminus\setA_s^{m\times n})=0$, then every $\matX$ with $\opP[\matX\in\setU]=1$ can be recovered from $k>s$ measurements with zero error probability since $\setU$ is countably $\mathscr{H}^{s}$-rectifiable.
\item  If  $\setU \subseteq \setA^{m\times n}_s$  is nonempty, Borel, and bounded, then every $\matX$ with $\opP[\matX\in\setU]=1$ can be recovered from $k>s$ measurements with arbitrarily small error probability by a $\beta$-H\"older continuous mapping with $\beta\in(0,1-s/k)$   since   $\setV=\overline{\setU}\subseteq \setA_s^{m\times n}$   is    $s$-rectifiable. 
 \end{enumerate}
\end{example}

We  proceed to particularizing our recovery thresholds for low-rank matrices. 

\begin{example}\label{exa:lowrank}
The set $\setM_r^{m\times n}$ of matrices in $\reals^{m\times n}$ that have rank no more than $r$ is a finite union of $\{\matzero\}$ and $C^1$-submanifolds of $\reals^{m\times n}$ of dimensions no more than $(m+n-r)r$. This follows by noting that the set of matrices in $\reals^{m\times n}$ of fixed rank $k$ is a $C^1$-submanifold of $\reals^{m\times n}$ of dimension $(m+n-k)k$
\cite[Ex. 5.30]{le00}, \cite[Ex. 1.7]{wells-differential}.
Application of \Cref{exasum,exaC1} in  \cref{LemmaExarec} therefore yields that  $\setM_r^{m\times n}$ is countably $(m+n-r)r$-rectifiable. 
Also note that $\setM_r^{m\times n}-\setM_r^{m\times n}= \setM_{2r}^{m\times n}$.

Similarly, for every  bounded subset  $\setU\subseteq \setM_r^{m\times n}$, $\overline{\setU}$ is $(m+n-r)r$-rectifiable.
This follows by first noting that $\overline{\setU}$  is compact in $\reals^{m\times n}$ 
and, therefore, 
the intersection of $\overline{\setU}$ with any of the finitely many $C^1$-submanifolds participating in $\setM_r^{m\times n}$ is a compact  subset of a $C^1$-submanifold. Hence, by   \cref{exaC12,exasum2} of \cref{LemmaExarec2},
$\overline{\setU}$  is $(m+n-r)r$-rectifiable.

We can therefore  apply the corresponding items of 
 \Cref{threct}  to obtain recovery thresholds for Lebesgue a.a.  
$((\veca_1 \dots \veca_k), (\vecb_1 \dots \vecb_k)) \in\reals^{m\times k}\times\reals^{n\times k}$ for the following sets: 
\begin{enumerate}[label=\roman*)]
\item If $\setU\subseteq  \setM_r^{m\times n}$ is nonempty, then every $\matX\in\setU$  can be recovered uniquely from $k>2(m+n-2r)r$ measurements since 
$\setU-\setU\subseteq  \setM_{2r}^{m\times n}  $ is countably $\mathscr{H}^{2(m+n-2r)r}$-rectifiable. 
  \item    If  $\setU\subseteq \setM_r^{m\times n}$ is nonempty and bounded, then every  $\matX\in\setU$  can be recovered uniquely from $k>2(m+n-2r)r$ measurements by a $\beta$-H\"older continuous mapping with $\beta\in(0,1-2(m+n-2r)r/k)$  since  $\setV=\overline{\setU-\setU}\subseteq \setM_{2r}^{m\times n}$ 
   is $(2(m+n-2r)r)$-rectifiable. 
\item If $\setU\subseteq \reals^{m\times n}$  is nonempty, Borel,  and satisfies  $\mathscr{H}^s(\setU\setminus\setM_r^{m\times n})=0$, then every $\matX$ with $\opP[\matX\in\setU]=1$ can be recovered from $k>(m+n-r)r$ measurements with zero error probability since $\setU$ is countably $\mathscr{H}^{(m+n-r)r}$-rectifiable.
\item  If  $\setU \subseteq \setM_r^{m\times n}$  is nonempty, Borel, and bounded, then every $\matX$ with $\opP[\matX\in\setU]=1$ can be recovered from $k>(m+n-r)r$ measurements with arbitrarily small error probability by a $\beta$-H\"older continuous mapping with $\beta\in(0,1-(m+n-r)r/k)$   since   $\setV=\overline{\setU}\subseteq \setM_r^{m\times n}$   is    $((m+n-r)r)$-rectifiable. 
 \end{enumerate}
\end{example}

We proceed with the development of our general theory by demonstrating that simple, albeit  relevant algebraic manipulations preserve rectifiability and hence allow
the direct statement of recovery thresholds in the spirit of \cref{threct} through application of the approach just described.

\begin{lemma}\label{lem:lipexm}
Let $\setU_i\subseteq \reals^{m\times n}$, for $i= 1,2$, 
and define 
\begin{enumerate}[label=\roman*)]
  \item $\setA =\{\matX\tp{\matX}:\matX\in\setU_1\}$,
  \item $\setA_{\times}=\{\matX_1\tp{\matX_2}:\matX_1\in \setU_1, \matX_2\in \setU_2\}$, \label{it:product}
  \item $\setA_{+}=\{\matX_1+\matX_2:\matX_1\in \setU_1, \matX_2\in \setU_2\}$,  \label{it:sum}
  \item $\setA_{\otimes}=\{\matX_1\otimes\matX_2:\matX_1\in \setU_1, \matX_2\in \setU_2\}$.
 \end{enumerate}
 If the sets $\setU_i$ are (countably) $s_i$-rectifiable, for $i= 1,2$,
then $\setA$ is (countably) $s_1$-rectifiable
and $\setA_{\times}, \setA_{+}$, and $\setA_{\otimes}$ are (countably) $(s_1+s_2)$-rectifiable.  
\end{lemma}
\begin{proof}
The mapping $\matX \mapsto \matX\tp{\matX}$ is continuously differentiable, and hence locally Lipschitz. 
Thus, by \cref{itemcrec} in \cref{lem:forexa},  the set $\setA$ is countably $s_1$-rectifiable for $\setU_1$ countably $s_1$-rectifiable,
and,  by \cref{itemrec} in \cref{lem:forexa},  the set $\setA$ is $s_1$-rectifiable for $s_1$-rectifiable $\setU_1$.
Similarly, all of the mappings 
$f_{\times}(\matX_1,\matX_2)  \mapsto \matX_1\tp{\matX_2}$,
$f_{+}(\matX_1,\matX_2)  \mapsto \matX_1+\matX_2$,
and 
$f_{\otimes}(\matX_1,\matX_2)  \mapsto \matX_1\otimes \matX_2$ 
are continuously differentiable, and thus locally Lipschitz.
Hence, by \cref{itemcrec} in \cref{lem:forexa},  the sets
$\setA_{\times}$,
$\setA_{+}$, and 
$\setA_{\otimes}$
are countably $(s_1+s_2)$-rectifiable
when the sets $\setU_i$ are countably $s_i$-rectifiable, for $i= 1,2$.
Likewise, by \cref{itemrec} in \cref{lem:forexa},  the sets
$\setA_{\times}$,
$\setA_{+}$, and 
$\setA_{\otimes}$
are $(s_1+s_2)$-rectifiable
when the sets $\setU_i$ are $s_i$-rectifiable, for $i= 1,2$.
\end{proof}

\cref{lem:lipexm} in combination with \cref{exa:sparse,exa:lowrank} 
immediately yields recovery thresholds for sums, products, and Kronecker products of sparse and low-rank matrices and covers, e.g.,   the structured matrices discussed in \cite{klopp2019structured}. A  concrete  example making use of    \cref{it:product} in \Cref{lem:lipexm} is the QR-decomposition of matrices with sparse R-components. 

\begin{example}\label{ex:qr}
  Let $m,n\in\naturals$ with $m\geq n$ and denote by $\setC_s^{m\times n}$  the set of all matrices  in $\reals^{m\times n}$ with $s$-sparse upper triangular matrix in their QR-decomposition, i.e, 
  \begin{align}
   \setC_s^{m\times n}=\{\matQ\matR: \matQ\in \setQ^{m\times m}, \matR\in \setR_s^{m\times n}\}, 
   \end{align}   
   where 
    $\setR_s^{m\times n}\subseteq\setA_s^{m\times n}$ designates  the set of all $s$-sparse upper triangular matrices  and 
   $\setQ^{m\times m}$ stands for   the set of orthogonal matrices in $\reals^{m\times m}$. Employing the same reasoning   as in  
   \cref{exa:sparse}, it follows that $\setR_s^{m\times n}$ is countably $s$-rectifiable. 
   Further, $\setQ^{m\times m}$ is a  compact $m(m-1/2)$-dimensional $C^1$-submanifold of $\reals^{m\times m}$ \cite[Section 1.3.1]{ch03} and thus, by \cref{exaC12} in \cref{LemmaExarec2}, $(m(m-1)/2)$-rectifiable. 
We therefore conclude that,  by \cref{it:product} in \cref{lem:lipexm}, $\setC_s^{m\times n}$ is countably $(m(m-1)/2+s)$-rectifiable. Further, thanks to \cref{lem:rect2U},   
  $\setC_s^{m\times n}-\setC_s^{m\times n}$ is countably $(m(m-1)+2s)$-rectifiable.   
  
  Now, consider a bounded subset  $\setU\subseteq \setC_s^{m\times n}$. 
  Then, \begin{align}
  \setU_2 = \big\{\matR \in \setR_s^{m\times n}: \exists \matQ\in \setQ^{m\times m}\ \text{with}\ \matQ\matR\in \setU \big\}\end{align} 
  is bounded because multiplication by $\matQ$ does not change the $2$-norm and 
   $\setU\subseteq \tilde{\setU} := \{\matQ\matR: \matQ\in \setQ^{m\times m}, \matR\in\overline{\setU_2}\}$.  
  Now,  $\setQ^{m\times m}$ is  $m(m-1)/2$-rectifiable, and using the same argumentation as in  
   \cref{exa:sparse} with  $\setR_s^{m\times n}$ in place of  $\setA^{m\times n}_s$, it follows  that $\overline{\setU_2}$ is $s$-rectifiable.  
  Thus,    $\tilde{\setU}$ is  $(m(m-1)/2+s)$-rectifiable owing to \cref{it:product} in \cref{lem:lipexm}. 
  Further, thanks to  \cref{lem:rect2U}, $\tilde{\setU}-\tilde{\setU}$ is  $(m(m-1)+2s)$-rectifiable.

  We can therefore  apply the corresponding items of 
 \Cref{threct}  to obtain recovery thresholds for Lebesgue a.a. 
$((\veca_1 \dots \veca_k), (\vecb_1 \dots \vecb_k)) \in\reals^{m\times k}\times\reals^{n\times k}$ for the following sets: 
\begin{enumerate}[label=\roman*)]
\item If $\setU\subseteq  \setC_s^{m\times n}$ is nonempty, then every $\matX\in\setU$  can be recovered uniquely from $k>m(m-1)+2s$ measurements since 
$\setU-\setU\subseteq \setC_s^{m\times n}-\setC_s^{m\times n}$ is countably $\mathscr{H}^{m(m-1)+2s}$-rectifiable. 
  \item    If  $\setU\subseteq \setC_s^{m\times n}$ is nonempty and bounded, then every  $\matX\in\setU$  can be recovered uniquely from $k>m(m-1)+2s$ measurements by a $\beta$-H\"older continuous mapping with $\beta\in(0,1-(m(m-1)+2s)/k)$  since  $\setV=\tilde{\setU}-\tilde{\setU}\subseteq \setC_s^{m\times n}-\setC_s^{m\times n}$ 
   is $(m(m-1)+2s)$-rectifiable. 
\item If $\setU\subseteq \reals^{m\times n}$  is nonempty, Borel,  and satisfies  $\mathscr{H}^s(\setU\setminus\setC_s^{m\times n})=0$, then every $\matX$ with $\opP[\matX\in\setU]=1$ can be recovered from $k>m(m-1)/2+s$ measurements with zero error probability since $\setU$ is countably $\mathscr{H}^{m(m-1)/2+s}$-rectifiable.
\item  If  $\setU \subseteq \setA^{m\times n}_s$  is nonempty, Borel, and bounded, then every $\matX$ with $\opP[\matX\in\setU]=1$ can be recovered from $k>m(m-1)/2+s$ measurements with arbitrarily small error probability by a $\beta$-H\"older continuous mapping with $\beta\in(0,1-(m(m-1)/2+s)/k)$   since   $\setV=\tilde {\setU}\subseteq \setC_s^{m\times n}$   is    $(m(m-1)/2+s)$-rectifiable. 
 \end{enumerate}
\end{example}

We finally note that since \cref{lem:forexa} holds for general $N\in\naturals$, \cref{lem:lipexm} is readily extended to sums, products, and Kronecker products of more than two matrices.    This extension
allows to deal, inter alia, with singular value decompositions and eigendecompositions in a manner akin to \cref{ex:qr}. Another interesting example, which can be worked out  using the same arguments as in \Cref{ex:qr} with \cref{it:sum} in \cref{lem:lipexm} in place of \cref{it:product} in \cref{lem:lipexm}, is the recovery of matrices that are  sums of  low-rank and  sparse matrices.

\section{Recurrent Iterated Function Systems}\label{sec:RIFS}

We now demonstrate how our theory can be applied to sets of fractal nature,  which  do not fall into the rich  class of rectifiable sets. 
Specifically, we investigate attractor sets of  recurrent iterated function systems defined as follows \cite{baelha89}.  
Let $\setK$ be a compact subset of $(\reals^m,\lVert\,\cdot\,\rVert_2)$ and fix $n\in\naturals$. 
 For $i=1,\dots,n$, let $w_i\colon \setK\to\setK$ be similitudes of contractivity $s_i\in[0,1)$, i.e., 
 \begin{align}
 \lVert w_i(x)-w_i(y) \rVert_2 = s_i \lVert x-y \rVert_2, \quad\text{for all $x,y\in\setK$ and $i=1,\dots,n$,}
 \end{align}
 and designate $\vecw=\tp{(w_1,\dots,w_n)}$. 
 Finally, let $\matP\in [0,1]^{n\times n}$  with entries $p_{i,j}$ in the $i$-th row and $j$-th column.   
  The triple $(\setK,\vecw,\matP)$ is referred to as a recurrent iterated function system.  In what follows, we assume that 
  $\matP$ is 
 \begin{enumerate}[label=\roman*)]
 \item  row-stochastic, i.e.,  
\begin{align}\label{eq:rowstochastic}
 \sum_{j=1}^n p_{i,j} =1, \quad \text{for $i\in\{1,\dots,n\}$, and}
 \end{align}
 \item irreducible, i.e., for every $i,j\in\{1,\dots,n\}$, there exist $i_1,\dots, i_m\in \{1,\dots,n\}$ such that $i_1=i$, $i_m=j$, and 
 \begin{align}\label{eq:irreducible}
 p_{i_1,i_2}\,p_{i_2,i_3}\dots p_{i_{m-1},i_m} >0. 
\end{align}
 \end{enumerate}
Further, define the connectivity matrix $\matC\in\{0,1\}^{n\times n}$ with entries $c_{i,j}$ in the $i$-th row and $j$-th column according to 
\begin{align}
c_{i,j}=
\begin{cases}
1,&\quad\text{if $p_{i,j}>0$}\\
0,&\quad\text{if $p_{i,j}=0$,}
\end{cases} \quad\text{for $i,j\in\{1,\dots,n\}$}
\end{align} 
and set $I(i)=\{j\in\{1,\dots,n\}: c_{i,j}=1\}$ for $i\in\{1,\dots,n\}$. Note that $\matP$ is irreducible if and only if  $\matC$ is irreducible.  
 For every recurrent iterated function system $(\setK,\vecw,\matP)$, 
 there exist unique nonempty compact sets $\setA_1,\dots,\setA_n\subseteq \setK$ satisfying \cite[Corollary 3.5]{baelha89}
 \begin{align} \label{eq:setAi}
\setA_i= \bigcup_{j\in I(i) } w_i(\setA_j),\quad\text{for $i=1,\dots,n$.}
 \end{align}
 The set 
 \begin{align}\label{eq:setUattractor}
 \setU=(\setA_1,\dots,\setA_n) \subseteq \reals^{m\times n}
 \end{align}
is called the attractor set of the recurrent iterated function system $(\setK,\vecw,\matP)$.     
We say that the sets $\setA_i$ in \cref{eq:setAi} are nonoverlapping if, for every $i\in\{1,\dots,n\}$, 
\begin{align}
\setA_j\cap\setA_k =\emptyset, \quad\text{for all $j,k\in I(i)$ with  $j\neq k$.}
\end{align}
To apply the recovery thresholds from  \cref{th1,th2} to   attractor sets $\setU$ according to  \cref{eq:setUattractor},  we  
need  the following dimension result: 
\begin{theorem}\label{thm:dA} \cite[Theorem 4.1]{baelha89}
Let $(\setK,\vecw,\matP)$ be a recurrent iterated function system with $\matP$ satisfying \cref{eq:rowstochastic} and \cref{eq:irreducible}. 
For every $t\in(0,\infty)$, define the diagonal matrix $\matS(t)=\diag(s_1^t,s_2^t,\dots,s_n^t)$, where $s_i$ is the contractivity of the similitude $w_i$, for $i=1,\dots,n$. Let $\setU=(\setA_1,\dots,\setA_n)$ be the attractor set of  $(\setK,\vecw,\matP)$ and suppose that the sets $\setA_1,\dots,\setA_n$ are  nonoverlapping.  Finally, let $d$ be the unique positive number such that $1$ is an eigenvalue of $\matS(d)\matC$ of maximum modulus (cf. \cite[Perron-Frobenius Theorem]{baelha89}). Then, it holds that  
\begin{align}
\max \big\{ \overline{\dim}_\mathrm{B}(\setA_i) : i=1,\dots,n \big\} = d.
\end{align}
\end{theorem}
One obtains the following immediate consequence: 
\begin{corollary}\label{cor:dA}
Under the assumptions of \Cref{thm:dA}, the attractor set\\ $\setU=(\setA_1,\dots,\setA_n)$ satisfies 
\begin{align}\label{eq:46}
\overline{\dim}_\mathrm{B}(\setU) \leq nd. 
\end{align}
\end{corollary}
\begin{proof}
Follows from \Cref{thm:dA} and the product formula \cite[Equation (7.9)]{fa14}.  
\end{proof}

We next present a simple example application of \Cref{thm:dA} and \Cref{cor:dA}, which can  easily be extended to higher dimensions.  
\begin{example} 
Let $s\in (0,1/2)$ and consider the similitudes $w_i\colon [0,1]^2\to[0,1]^2$  of contractivity $s$ defined as $w_i(\vecx) =s\,\vecx+\vecb_i$, for $i=1,\dots,4$, where 
$\vecb_1=\tp{(0,0)}$,  $\vecb_2=\tp{(1-s,0)}$, $\vecb_3=\tp{(0,1-s)}$, and $\vecb_4=\tp{(1-s,1-s)}$.  
Now,  let $\matP\in [0,1]^{4\times 4}$ be a row-stochastic matrix  
and  suppose that $p_{i,i}=0$, for $i=1,\dots, 4$, and $p_{i,j}>0$, for $i,j\in \{1,\dots,4\}$ with $i\neq j$. 
Further, let $\setA_1,\dots\setA_4$ be as in \cref{eq:setAi} and $\setU$ as in 
\cref{eq:setUattractor}. 
By construction, the sets $\setA_i$ are nonoverlapping as  \cref{eq:setAi} implies 
\begin{align}
\setA_i \subseteq w_i([0,1]^2)\quad\text{ for $i=1,\dots, 4$}
\end{align}
and the $w_i's$ have  pairwise disjoint  codomains.  
 Next,  note that, owing to \cite[Theorem 1.3.22]{hojo13}, the characteristic polynomial  of the all-ones matrix 
\begin{align}
\matJ=
\begin{pmatrix}
1&1&1&1\\1&1&1&1\\1&1&1&1\\1&1&1&1
\end{pmatrix}
\end{align}
is given by 
 \begin{align}\label{eq:pJ}
 p_\matJ(x)=\det(\matJ-x\matI)=(x-4)x^3,
 \end{align}
  where 
$\matI=\diag(1,\,1,\,1,\,1)$. We conclude that the characteristic polynomial  of the matrix $\matS(t)\matC$ equals 
\begin{align}
p_{\matS(t)\matC}(x)
&=p_{s^{t}\matC}(x)\\
&=s^{4t}\det(\matC-s^{-t}x\matI)\\
&=s^{4t}\det(\matJ-(s^{-t}x+1)\matI)\label{eq:useCI}\\
&= s^{4t} p_\matJ(s^{-t}x+1)\\
&=s^{4t} (s^{-t}x-3)(s^{-t}x+1)^3\label{eq:pSC},
\end{align}  
where \cref{eq:useCI} follows from $\matC=\matJ-\matI$ and in \cref{eq:pSC} we applied \cref{eq:pJ}. 
Hence, the eigenvalue of maximum modulus of $\matS(t)\matC$ is $\lambda_{\mathrm{max}}=3s^{t}$. Setting  $t=\log (1/3)/\log(s)$ therefore yields $\lambda_{\mathrm{max}}=1$ so that 
\begin{align}
\max\big\{ \overline{\dim}_\mathrm{B}(\setA_i) : i=1,\dots,4 \big\} = \frac{\log (1/3)}{\log(s)} 
\end{align}  
owing to \Cref{thm:dA}  
and hence 
\begin{align}
\overline{\dim}_\mathrm{B}(\setU)\leq \frac{4\log (1/3)}{\log(s)} 
\end{align}
thanks to \Cref{cor:dA}. 
\end{example}

The upper bound in \cref{eq:46}  now leads to the following recovery thresholds: 

\begin{theorem}\label{th1RIRF}(Recovery of matrices taking  values in attractor sets)
Let $\setU$ be the attractor set of a recurrent iterated function system satisfying the assumptions of \Cref{thm:dA}. Then, the following statements hold. 
\begin{enumerate}[label=\roman*)]
\item\label{en:RIFS1}
For $k > 2nd$ and 
Lebesgue a.a. 
$((\veca_1 \dots \veca_k), (\vecb_1 \dots \vecb_k))$ 
$\in \reals^{m\times k}\times \reals^{n\times k}$, every  
$\matX\in\setU$ can be recovered uniquely from the rank-1 measurements
\begin{align}
\tp{(\tp{\veca_1}\matX\vecb_1\ \dots\ \tp{\veca_k}\matX\vecb_k)}.
\end{align}
\item\label{en:RIFS2}
Let  $\beta \in \big(0, \big(1-\frac{2nd}{k}\big)\big)$ with $k>2nd$.
Then, recovery in \cref{en:RIFS1} can be accomplished by a $\beta$-H\"older continuous mapping  $g$. 
\item\label{en:RIFS3}
Let $\rmatX$ be a random matrix satisfying $\opP[\rmatX\in\setU]=1$.
Then, for Lebesgue a.a.   
$((\veca_1 \dots \veca_k), (\vecb_1 \dots \vecb_k))  \in\reals^{m\times k}\times \reals^{n\times k}$,  
there exists 
a Borel-measurable mapping $g\colon \reals^k\to\reals^{m\times n}$ satisfying  
\begin{align}
\opP\Big[g\Big(\tp{(\tp{\veca_1}\rmatX\vecb_1\ \dots\ \tp{\veca_k}\rmatX\vecb_k)}\Big)\neq\rmatX\Big]=0
\end{align}
provided that  $k > nd$. 
\item\label{en:RIFS4}
Let $\rmatX$ be a random matrix satisfying $\opP[\rmatX\in\setU]=1$, fix $\varepsilon>0$, and let $\beta \in \big(0,\big(1-\frac{nd}{k}\big)\big)$ with $k > nd$. Then, for
Lebesgue a.a.  
 $((\veca_1 \dots \veca_k), (\vecb_1 \dots \vecb_k))  \in\reals^{m\times k}\times \reals^{n\times k}$,  
there exists 
a $\beta$-H\"older continuous mapping $g\colon \reals^k\to\reals^{m\times n}$ satisfying  
\begin{align}
\opP\Big[g\Big((\tp{{\tp{\veca_1}}\rmatX\vecb_1\ \dots\ \tp{\veca_k}\rmatX\vecb_k)}\Big)\neq\rmatX\Big] \leq \varepsilon.
\end{align}
\end{enumerate}
\end{theorem}
\begin{proof}
With  $\overline{\dim}_\mathrm{B}(\setU) \leq nd$ from  \Cref{cor:dA},  we have  
\begin{align}
\dim_\mathrm{H}(\setU-\setU) \leq \overline{\dim}_\mathrm{B}(\setU-\setU)\leq 2 \overline{\dim}_\mathrm{B}(\setU) \leq 2nd
\end{align} 
thanks  to  \Cref{lem:rect2U}. The statements in \cref{en:RIFS1}--\cref{en:RIFS4} now follow readily  from the corresponding parts of \Cref{th1,th2}.  
\end{proof}

\section{Proof of \cref{th1}}\label{sec:proofth1}
\Cref{th1item1} is by linearity of the mapping defined in \eqref{eq:linearity1}--\eqref{eq:linearity2} combined with \cref{prp:ns} applied to 
the set $\setU-\setU$. 
 
The proof of \cref{th1item2}  follows along the same lines as  the proof of \cite[Theorem II.1]{albole19}. We therefore present a proof sketch only. 
First, note that
by \cite[Lemma 2.3]{bagusp19}, there exist compact sets $\setU_i\subseteq \setU$, $i\in\naturals$, such that $\opP[\rmatX\in\setV]=1$, where 
\begin{align}\label{eq:VUI}
\setV=\bigcup_{i\in\naturals}\setU_i. 
\end{align} 
Next, consider the ``encoder'' mapping%
\footnote{$\matA$ and $\matB$ denote the matrices with $\veca_i$ and $\vecb_i$, respectively, in their $i$-th column.} 
\begin{align}
e\colon \reals^{m\times k}\times \reals^{n\times k} \times \reals^{m\times n}&\to \reals^k\\
\big(\matA, \matB,  \matV \big)&\mapsto  \tp{(\tp{\veca_1}\matV\vecb_1\ \dots\ \tp{\veca_k}\matV\vecb_k)}.  
\end{align}
With the decomposition of $\setV$ in \cref{eq:VUI},   argumentation as in \cite[Section V.A]{albole19} (with the mapping $\lVert\vecy-\matA\vecv\rVert_2$ 
in \cite[(139)--(140)]{albole19} replaced by  $\lVert\vecy-e(\matA,\matB,\matV)\rVert_2$) 
implies the existence of a  measurable mapping  
\begin{align}
\hat g \colon\reals^{m\times k}\times \reals^{n\times k}\times\reals^k&\to\reals^{m\times n}\\
\matA\times\matB\times\vecy&\mapsto \matX
\end{align}
such that
\begin{align}\label{eq:gprop}
e\big(\matA, \matB, \hat g(\matA, \matB,  \vecy )\big) = \vecy,\, \text{for all $\matA\in\reals^{m\times k}$, $\matB\in\reals^{n\times k}$ and $\vecy\in e(\{\matA\} \times \{\matB\} \times \setV)$.} 
\end{align}
Moreover, the mapping $\hat g $ is guaranteed to deliver an $\matX \in \setV$ that is consistent if at least one such consistent $\matX \in \setV$ exists, otherwise an error is declared by delivering an error symbol not contained in $\setV$. 
Next, for every $\matA\in\reals^{m\times k}$ and $\matB\in\reals^{n\times k}$, 
let  $p_\mathrm{e}(\matA,\matB)$ denote the probability of error defined as 
\begin{align}
p_\mathrm{e}(\matA,\matB)
& =\opP[\hat g(\matA, \matB, e(\matA, \matB,  \rmatX )) \neq\rmatX]. 
\label{eq:peab}
\end{align} 
We now   show that $p_\mathrm{e}(\matA, \matB)=0$ for Lebesgue a.a. $(\matA, \matB)$.  
We  have 
\begin{align}
& \int p_\mathrm{e}(\matA,\matB) \,\mathrm d\lebmeasure(\matA,\matB) \\ 
&=\opE\big[\lebmeasure\big(\big\{(\matA,\matB): 
g\big(\matA,\matB, e(\matA, \matB,  \rmatX ) \big)\neq\rmatX
\big\}\big)\ind{\setV}(\rmatX)\big]\label{eq:useuniona}\\
&\leq\opE\big[\lebmeasure\big(\big\{(\matA,\matB):\{\widetilde{\matV}\in \setV_\rmatX: e(\matA, \matB,  \widetilde{\matV} ) =\matzero\} \neq\{\matzero\}\big\}\big)\big], \label{eq:finaltheorem}
\end{align}
where 
 \cref{eq:useuniona} follows from Fubini's theorem \cite[Theorem 1.14]{ma99} together with  $\opP[\rmatX\in\setV]=1$, and 
in \cref{eq:finaltheorem},
we set  $\setV_\matX=\{\matV-\matX:\matV \in \setV\}$ 
and used the fact that, by \cref{eq:gprop}, $\matV:=\hat g\big(\matA,\matB, e(\matA, \matB,  \matX )\big)\neq\matX$ with $\matX\in\setV$ implies that  $\matV\in\setV\!\setminus\!\{\matX\}$ with 
\begin{align}
e(\matA, \matB,  \matX )=e(\matA, \matB,  \matV), 
\end{align}
i.e., $e(\matA, \matB,  \matV-\matX )=\veczero$. 
Finally, since 
$\mathscr{H}^k(\setV_\matX)=\mathscr{H}^k(\setV)$ by the translation invariance of $\mathscr{H}^k$ and $\mathscr{H}^k(\setV)=0$ as a consequence of $\setV\subseteq \setU$ and $\dim_\mathrm{H}(\setU)<k$, the expectation in \cref{eq:finaltheorem} is equal to zero owing to \cref{prp:ns}. Finally, for fixed $\matA,\matB$, set $g=\hat g(\matA,\matB,\cdot)$.\qed

\section{Proof of \cref{prp:ns}}\label{prp:nsproof} 
For every $j\in\naturals$, set 
\begin{align}
\setA(j)&=\underbrace{\setB_{m}(\veczero,j)\times\dots\times \setB_{m}(\matzero,j)}_{k\ \text{times}}\quad \text{and}\label{eq:Al}\\
\setB(j)&=\underbrace{\setB_{n}(\veczero,j)\times\dots\times \setB_{n}(\matzero,j)}_{k\ \text{times}}. \label{eq:Bl}
\end{align} 
By countable subadditivity of  Lebesgue measure, it suffices to show that 
\begin{align}\label{eq:proptoshow2}
\big\{\matX\in\setU\!\setminus\!\{\matzero\}: \tp{(\tp{\veca_1}\matX\vecb_1\ \dots\ \tp{\veca_k}\matX\vecb_k)}=\veczero\big\}=\emptyset,
\end{align}
for  Lebesgue  a.a.  $((\veca_1 \dots \veca_k), (\vecb_1 \dots \vecb_k))  \in\setA(j)\times\setB(j)$ and all $j\in\naturals$.
By \cref{lem:probzero2} below, \cref{eq:proptoshow2} then holds, for all $j\in\naturals$, with probability one if the deterministic matrices $\big((\veca_1 \dots \veca_k), (\vecb_1 \dots \vecb_k)\big)  \in\setA(j)\times\setB(j)$ are replaced by independent random matrices with columns $\rveca_i$, $i=1,\dots,k$, independent and uniformly distributed on $\setB_m(\veczero,j)$,  and columns $\rvecb_i$, $i=1,\dots,k$, independent and uniformly distributed on $\setB_n(\veczero,j)$. 
By countable subadditivity of  Lebesgue measure, this finally
implies that  \cref{eq:proptoshow} can be violated only on a set of Lebesgue measure zero, which concludes the proof.
 
\begin{lemma}\label{lem:probzero2}
Let $s>0$ and take $\rmatA=(\rveca_1\dots\rveca_k)$ and $\rmatB=(\rvecb_1\dots\rvecb_k)$ to be  independent random matrices with columns $\rveca_i$, $i=1,\dots,k$, independent and uniformly distributed on $\setB_m(\veczero,s)$,  and columns $\rvecb_i$, $i=1,\dots,k$, independent and uniformly distributed on $\setB_n(\veczero,s)$. 
Consider $\setU\subseteq \reals^{m\times n}$ with   $\dim_\mathrm{H}(\setU)<k$.  
Then, 
\begin{align}
P:=\opP\big[\exists \matX\in\setU\!\setminus\!\{\veczero\}:\tp{(\tp{\rveca_1}\matX\rvecb_1\ \dots\ \tp{\rveca_k}\matX\rvecb_k)}=\matzero\big]=0. 
\end{align}
\end{lemma}
\begin{proof}
For every $L\in\naturals$, let 
\begin{align}\label{eq:defUl}
\setU_{L}=\Big\{\matX\in\setU:  \tfrac{1}{L}< \sigma_1(\matX)<L\Big\}  
\end{align}
and set 
\begin{align}
P_{L} =\opP\big[\exists \matX\in\setU_{L} : \tp{(\tp{\rveca_1}\matX\rvecb_1\ \dots\ \tp{\rveca_k}\matX\rvecb_k)}=\matzero\big]. 
\label{PLr}
\end{align}
By the union bound, we have  
\begin{align}\label{eq:seriesprob}
P\leq
\sum_{L\in\naturals}P_{L}.  
\end{align}
We now fix $L\in\naturals$ arbitrarily and  prove  that $P_{L}=0$.  
Let  $\kappa=(k+\dim_\mathrm{H}(\setU))/2$. 
As $k>\dim_\mathrm{H}(\setU)$ by assumption, it follows that $\dim_\mathrm{H}(\setU)<\kappa<k$.  
In particular, $\kappa>\dim_\mathrm{H}(\setU)$ implies, by \cite[Equation (3.11)]{fa14}, that $\mathscr{H}^\kappa(\setU)=0$  and in turn  $\mathscr{H}^\kappa(\setU_L)=0$ 
by monotonicity of $\mathscr{H}^\kappa$.  
Thus, 
 $\mathscr{M}^\kappa(\setU_L)=0$ by \cite[Section~3.4]{fa14}, where the measure $\mathscr{M}^\kappa$ is defined according to 
\begin{align}
\mathscr{M}^{\kappa}(\setV) =\lim_{d\to 0}\mathscr{M}_d^{\kappa}(\setV)
\end{align} 
with 
\begin{align}
\mathscr{M}_d^{\kappa}(\setV)
&=\inf\Bigg\{\sum_{i\in\naturals}\varepsilon_i^\kappa : \setV\subseteq\bigcup_{i\in\naturals}\setB_{m\times n}\big(\matX_i,\tfrac{\varepsilon_i}{2}\big)\Bigg\}, \quad\text{for  all $d>0$},   
\end{align}
where the infimum is taken over all possible ball centers  $\matX_i\in\reals^{m\times n}$ and radii $\varepsilon_i\in (0,d)$, $i\in\naturals$. 
Since $\mathscr{M}_d^{\kappa}(\setU_L)$ is nonnegative and monotonically nondecreasing as $d\to 0$, $\mathscr{M}^\kappa(\setU_L)=0$ implies $\mathscr{M}_d^{\kappa}(\setU_L)=0$, for all $d>0$. 
Now, fix $d >0$ and  $\varepsilon \in (0,(\sqrt{k}L)^{-\kappa})$   arbitrarily.  
As $\mathscr{M}_d^{\kappa}(\setU_L)=0$, there must 
exist ball centers $\matX_i\in\reals^{m\times n}$, $i\in\naturals$, and radii $\varepsilon_i$, $i\in\naturals$,
such that 
\begin{align}
\setU_L\subseteq \bigcup_{i\in\naturals}\setB_{m\times n}\big(\matX_i,\tfrac{\varepsilon_i}{2}\big)\label{eq:existcenters1}
\end{align}
and
\begin{align}\label{eq:epsilon}
\sum_{i\in\naturals}\varepsilon_i^\kappa <\varepsilon. 
\end{align}
As \cref{eq:existcenters1,eq:epsilon} continue to hold upon removal of all $i$ that satisfy 
\begin{equation}\setU_L\cap\setB_{m\times n}\big(\matX_i,\tfrac{\varepsilon_i}{2}\big)=\emptyset,\end{equation}
we can assume, w.l.o.g., that
\begin{align}
\setU_L\cap\setB_{m\times n}\big(\matX_i,\tfrac{\varepsilon_i}{2}\big)\neq \emptyset, \quad\text{for all $i\in\naturals$}. \label{eq:existcenters2}
\end{align}
By doubling the radius, we can further construct a covering that has all its ball centers in $\setU$.
Concretely, by \cref{eq:existcenters2}, for every $i\in \naturals$, there exists $\matY_i\in \setU_L\cap\setB_{m\times n}(\matX_i,\varepsilon_i/2)$, and we have
$\setB_{m\times n}(\matX_i,\varepsilon_i/2) \subseteq \setB_{m\times n}(\matY_i,\varepsilon_i)$.
Thus, by \cref{eq:existcenters1},
\begin{align}\label{eq:UL2}
\setU_L\subseteq \bigcup_{i\in\naturals}\setB_{m\times n}(\matY_i,\varepsilon_i).
\end{align}
With the definition of $\setU_L$ in \cref{eq:defUl}, we now obtain for the shifted ball centers
\begin{align}\label{eq:boundYi}
\tfrac{1}{L}<\sigma_1(\matY_i) <L, \quad\text{for all $i\in\naturals$}.
\end{align}
A union bound argument applied to \cref{PLr} in combination with \cref{eq:UL2} yields
\begin{align}\label{eq:UBB}
P_{L} &\leq \sum_{i\in\naturals}\opP\big[\exists \matX\in\setB_{m\times n}(\matY_i,\varepsilon_i) : \tp{(\tp{\rveca_1}\matX\rvecb_1\ \dots\ \tp{\rveca_k}\matX\rvecb_k)}=\matzero\big]. 
\end{align}
To bound the individual probabilities on the right-hand side of \cref{eq:UBB}, we proceed as follows.
Suppose that $\matX\in\setB_{m\times n}(\matY_i,\varepsilon_i)$ for some $i\in\naturals$. 
Then, we have 
\begin{align}
&\lVert\tp{(\tp{\rveca_1}\matY_i\rvecb_1\ \dots\ \tp{\rveca_k}\matY_i\rvecb_k)}\rVert_2\label{eq:boundLS1}\\ 
&\leq \lVert\tp{(\tp{\rveca_1}(\matX-\matY_i)\rvecb_1\ \dots\ \tp{\rveca_k}(\matX-\matY_i)\rvecb_k)}\rVert_2
+\lVert\tp{(\tp{\rveca_1}\matX\rvecb_1\ \dots\ \tp{\rveca_k}\matX\rvecb_k)}\rVert_2\\
&\leq\sqrt{\sum_{j=1}^k \lVert\rveca_j\rVert_2^2\lVert\matX-\matY_i\rVert_2^2\lVert\rvecb_j\rVert_2^2}
+\lVert\tp{(\tp{\rveca_1}\matX\rvecb_1\ \dots\ \tp{\rveca_k}\matX\rvecb_k)} \rVert_2\\
&\leq s^2\sqrt{k}\varepsilon_i +\lVert\tp{(\tp{\rveca_1}\matX\rvecb_1\ \dots\ \tp{\rveca_k}\matX\rvecb_k)} \rVert_2, \label{eq:boundLS}
\end{align}
where in \cref{eq:boundLS} we used that $\rveca_j$ and $\rvecb_j$  are uniformly distributed on $\setB_m(\veczero,s)$ and  $\setB_n(\veczero,s)$, respectively,   and $\matX\in \setB_{m\times n}(\matY_i,\varepsilon_i)$ by assumption. 
Thus, the event that there exists $\matX\in\setB_{m\times n}(\matY_i,\varepsilon_i)$ satisfying
\begin{align}
\tp{(\tp{\rveca_1}\matX\rvecb_1\ \dots\ \tp{\rveca_k}\matX\rvecb_k)}=\matzero 
\end{align}
implies  that
$\lVert\tp{(\tp{\rveca_1}\matY_i\rvecb_1\ \dots\ \tp{\rveca_k}\matY_i\rvecb_k)}\rVert_2\ \leq s^2\sqrt{k}\varepsilon_i$. 
Hence, we can further upper-bound $P_{L}$ according to
\begin{align}
P_{L} 
&\leq \sum_{i\in\naturals}\opP\big[\rVert\tp{(\tp{\rveca_1}\matY_i\rvecb_1\ \dots\ \tp{\rveca_k}\matY_i\rvecb_k)}\rVert_2\ \leq s^2\sqrt{k}\varepsilon_i\big]\label{eq:final0}\\
&\leq  k^{\frac{k}{2}}\sum_{i\in\naturals} \varepsilon_i^{k}\frac{2^{\frac{k(m+n)}{2}}}{\sigma_1(\matY_i)^k  }\Bigg(1+\log \Bigg(\frac{\sigma_1(\matY_i)}{\sqrt{k}\varepsilon_i}\Bigg)\Bigg)^k\label{eq:final1}
\\
&\leq C \sum_{i\in\naturals} \varepsilon_i^k\Bigg(1+\log \Bigg(\frac{L}{\sqrt{k}\varepsilon_i}\Bigg)\Bigg)^k\\
&= C \sum_{i\in\naturals} \varepsilon_i^\kappa  \varepsilon_i^{k-\kappa}\Bigg(1+\log \Bigg(\frac{L}{\sqrt{k}\varepsilon_i}\Bigg)\Bigg)^k\label{eq:final2}
\end{align}
with 
\begin{align}
C= 2^{\tfrac{k(m+n)}{2}}\big(L\sqrt{k}\big)^k,
\end{align}
where \cref{eq:final1} is by  \cref{lem:com1} below for $\delta=s^2\sqrt{k}\varepsilon_i$ and $\matX=\matY_i$ upon noting that 
\begin{align}
s^2\sqrt{k}\varepsilon_i 
&< s^2\sqrt{k}\varepsilon^{1/\kappa}\label{eq:bounde1}\\ 
&< \frac{s^2}{L}\label{eq:bounde2}\\
&<  \sigma_1(\matY_i)s^2,\quad\text{for all $i\in\naturals$}. \label{eq:bounde3} 
\end{align} 
Here, \cref{eq:bounde1} is by \cref{eq:epsilon}, in \cref{eq:bounde2} we used $\varepsilon<(\sqrt{k}L)^{-\kappa}$ which holds by assumption, and 
\cref{eq:bounde3}  follows from \cref{eq:boundYi}. 
As  the $\log$-term in  \cref{eq:final2} is dominated by $\varepsilon_i^{k-\kappa}$ for $\varepsilon_i \to 0$ thanks to $k>\kappa$,  \cref{eq:final2}   tends to zero for  $\varepsilon\to 0$ by \cref{eq:epsilon}.
We can therefore conclude that $P_L=0$, which, as $L$ was arbitrary, by \cref{eq:seriesprob}, implies $P=0$.  
\end{proof}

\begin{lemma}\label{lem:com1} 
Let $\rmatA=(\rveca_1\dots \rveca_k)$ and $\rmatB=(\rvecb_1\dots \rvecb_k)$ be independent random matrices, with columns $\rveca_i$, $i=1,\dots,k$, independent and uniformly distributed on $\setB_m(\veczero,s)$ and  columns $\rvecb_i$, $i=1,\dots,k$, independent and uniformly distributed on $\setB_n(\veczero,s)$. 
Suppose that $\matX\in\reals^{m\times n}\!\setminus\!\{\matzero\}$.
Then, we have  
\begin{align}
\opP\big[\big\lVert\tp{(\tp{\rveca_1}\matX\rvecb_1\ \dots\ \tp{\rveca_k}\matX\rvecb_k)}\big\rVert_2\leq \delta\big]\leq
\delta^{k}\frac{2^{\frac{k(m+n)}{2}}}{\sigma_1(\matX)^k s^{2k} }\Bigg(1+\log \Bigg(\frac{s^2\sigma_1(\matX)}{\delta}\Bigg)\Bigg)^k,
\end{align}
for all $\delta\leq \sigma_1(\matX)s^2$. 
\end{lemma}
\begin{proof} 
We have 
\begin{align}
&\opP\big[\big\lVert\tp{(\tp{\rveca_1}\matX\rvecb_1 \dots \tp{\rveca_k}\matX\rvecb_k)}\big\rVert_2\leq \delta\big]\\
&=\opP\big[\sum_{i=1}^k(\tp{\rveca}_i\matX\rvecb_i)^2\leq \delta^2\big]\\
&\leq \opP\big[|\tp{\rveca}_i\matX\rvecb_i|\leq \delta,\ \text{for}\ i=1,\dots,k\big]\\
&=
\opP\big[|\tp{\rveca}\matX\rvecb|\leq \delta\big]^{k} \label{eq:unioncom1}\\
&\leq\delta^{k}\frac{2^{\frac{k(m+n)}{2}}}{\sigma_1(\matX)^k s^{2k} }\Bigg(1+\log \Bigg(\frac{s^2\sigma_1(\matX)}{\delta}\Bigg)\Bigg)^k,
\label{eq:applycom1}
\end{align}
where in \cref{eq:unioncom1}  $\rveca$ and  $\rvecb$ are independent with $\rveca$ uniformly distributed on $\setB_m(\veczero,s)$ and $\rvecb$ uniformly distributed on $\setB_n(\veczero,s)$ and, therefore, we can apply \cref{lem:com} below to obtain  \cref{eq:applycom1}.
\end{proof}

\begin{lemma}\cite[Lemma 17]{lilebr17}\footnote{Since the assumption $\delta\leq \sigma_1(\matX)s^2$ is missing in \cite[Lemma 17]{lilebr17} we  present the proof of the lemma for completeness. A slightly weaker form of this result was first presented in \cite[Lemma 5]{ristbo15}.  
}\label{lem:com}
Let $\rveca$ and $\rvecb$ be independent random vectors, with  $\rveca$ uniformly distributed on $\setB_m(\veczero,s)$ and 
$\rvecb$ uniformly distributed on $\setB_n(\veczero,s)$, and  
suppose that $\matX\in\reals^{m\times n}\!\setminus\!\{\matzero\}$.  
Then, we have  
\begin{align}
\opP[|\tp{\rveca}\matX\rvecb|\leq \delta]
&\leq 
\delta\frac{D_{m,n}}{\sigma_1(\matX)s^2 }\Bigg(1+\log \Bigg(\frac{s^2\sigma_1(\matX)}{\delta}\Bigg)\Bigg), \quad\text{for all $\delta\leq \sigma_1(\matX)s^2$,}
\end{align}
where\footnote{We use the convention $V(0,s)=1$, for all $s\in \reals_{+}$.}
\begin{align}
D_{m,n}
&=\frac{4V(n-1,1)V(m-1,1)}{V(m,1)V(n,1)}\\
&\leq 2^{\frac{m+n}{2}}. 
\label{eq:D}
\end{align}
\end{lemma}
\begin{proof} 
We start by applying Fubini's Theorem \cite[Theorem 1.14]{ma99} and rewriting
\begin{align}
\opP[|\tp{\rveca}\matX\rvecb|\leq \delta]
&=\frac{1}{V(m,s)V(n,s)}\int_{\setB_m(\veczero,s)}h(\veca)\,\mathrm d\lebmeasure(\veca)\label{eq:comstep1}
\end{align}
with 
\begin{align}
h(\veca)
&=\int_{\setB_n(\veczero,s)}\ind{\{\vecb\in\reals^n : |\tp{\veca}\matX\vecb|\leq \delta\}}(\vecb)\,\mathrm d\lebmeasure(\vecb).\label{eq:h1}
\end{align}
Let $\matX=\matU\matSigma\matV$ be a singular value decomposition of $\matX$, where 
$\matU\in\reals^{m\times m}$ and 
$\matV\in\reals^{n\times n}$ are  orthogonal matrices, and 
\begin{align}
\matSigma=
\begin{pmatrix}
\matD&\matzero\\
\matzero&\matzero
\end{pmatrix}
\in\reals^{m\times n}
\end{align}
with $\matD=\diag(\sigma_1(\matX)\ldots\sigma_r(\matX))$ and  $r=\rank(\matX)$.  
Using the fact that Lebesgue measure on  $\setB_m(\veczero,s)$ and $\setB_n(\veczero,s)$ is invariant under rotations, we can write
\begin{align}
\opP[|\tp{\rveca}\matX\rvecb|\leq \delta]
&=\frac{1}{V(m,s)V(n,s)}\int_{\setB_m(\veczero,s)}h(\matU\veca)\,\mathrm d\lebmeasure(\veca)\label{eq:comstep2}
\end{align}
and 
\begin{align}
h(\matU\veca)
&=\int_{\setB_n(\veczero,s)}\ind{\{\vecb\in\reals^n : |\tp{\veca}\matSigma\vecb|\leq \delta\}}(\vecb)\,\mathrm d\lebmeasure(\vecb).\label{eq:h2}
\end{align} 
We now make the dependence on the largest eigenvalue $\sigma_1(\matX)$ explicit 
according to
\begin{align}
h(\matU\veca)
& = \int_{\setB_{n-1}(\veczero,s)}\int_{-s}^s \ind{\{b_1\in\reals: |\tp{\veca}\matSigma\vecb|\leq \delta, \lVert\vecb\rVert \leq s\}}(b_1)\,\mathrm d\lebmeasure(b_1)\,\mathrm d\lebmeasure\big(\tp{(b_2\dots b_n)}\big) \\
&\leq \int_{\setB_{n-1}(\veczero,s)} g\big(\tp{(b_2\dots b_n)}\big)\,\mathrm d\lebmeasure\big(\tp{(b_2\dots b_n)}\big)\label{eq:boundhUa}
\end{align}
with 
\begin{align}
g\big(\tp{(b_2\dots b_n)}\big)
&= \min\Bigg\{2s,\int_{-\infty}^\infty \ind{\{b_1\in\reals: |\sum_{i=1}^r\sigma_i(\matX)a_ib_i|\leq \delta\}}(b_1)\,\mathrm d\lebmeasure(b_1)\Bigg\}\label{eq:hua1}\\
&
= \min\Bigg\{2s,\int_{-\infty}^\infty \ind{\{b_1\in\reals: |\sigma_1(\matX)a_1b_1|\leq \delta\}}(b_1)\,\mathrm d\lebmeasure(b_1)\Bigg\}\label{eq:hua2}\\
&=  2\min\Bigg\{s,\frac{\delta}{\sigma_1(\matX)|a_1|}\Bigg\}. \label{eq:hua3}
\end{align} 
Using  \cref{eq:boundhUa} and \cref{eq:hua1}--\cref{eq:hua3} in \cref{eq:comstep2}, we obtain
\begin{align}
\opP[|\tp{\rveca}\matX\rvecb|\leq \delta]
&\leq \frac{D_{m,n}}{s^2}\int_{0}^s \min\Bigg\{s,\frac{\delta}{\sigma_1(\matX) a_1}\Bigg\}\,\mathrm d\lebmeasure(a_1)\\
&=\frac{\delta D_{m,n}}{\sigma_1(\matX) s^2}\Bigg(1+\log \Bigg(\frac{s^2\sigma_1(\matX)}{\delta}\Bigg)\Bigg),
\end{align}
for all $\delta\leq \sigma_1(\matX)s^2$.
The upper bound on $D_{m,n}$ follows from $2^{k/2}<V(k,1)<2^k$, for all $k\in\naturals$.
\end{proof}

\section{Proof of \cref{th2}}\label{sec:proofth2}
We first prove \cref{th2item1}. 
Consider the mapping 
\begin{align}
h\colon \setU&\to h(\setU)\subseteq\reals^k\\
\matX&\mapsto \tp{(\tp{\veca_1}\matX\vecb_1\ \dots\ \tp{\veca_k}\matX\vecb_k)}. 
\end{align} 
Application of \cref{prp:nsreg} to $\setU-\setU$ establishes the existence of a $c>0$ such that 
\begin{align}
\lVert h(\matU-\matV)\rVert_2  \geq c\lVert\matU-\matV\rVert_2^{1/\beta},    \quad\text{for all $\matU,\matV\in\setU$,}  
\end{align}
for Lebesgue a.a. $((\veca_1 \dots \veca_k), (\vecb_1 \dots \vecb_k))  \in\reals^{m\times k}\times \reals^{n\times k}$.
Hence, by \cite[Lemma 2]{striagbo15}, $h$ admits a  $\beta$-H\"older continuous inverse $h^{-1}\colon h(\setU)\to\setU$, which  can be extended to the desired   $\beta$-H\"older continuous mapping $g$ on $\reals^k$ owing to  \cite[Theorem 1, Item ii)]{mi70}. 

The proof of \cref{th2item2} follows along the same lines as that of \cite[Theorem 2]{striagbo15}. We therefore present a proof sketch only. 
By \cite[Proposition 2.6]{fa14}, we can assume, w.l.o.g., that $\setU$ is compact.   
Consider the sets $\setA, \setA_j\subseteq \reals^{m\times k}\times \reals^{n\times k}\times  \reals^{m\times n}$ defined according to 
\begin{align}\label{eq:defA}
\setA&=\Bigg\{\big(\matA,\matB,\matX\big) : \inf \Bigg\{\frac{\lVert\tp{(\tp{\veca_1}\matU\vecb_1\ \dots\ \tp{\veca_k}\matU\vecb_k)}\rVert_2}{\lVert\matU\rVert_2^{1/\beta}}: \matU\in\setU_\matX\!\setminus\!\{\matzero\} \Bigg\}=0  \Bigg\}  
\end{align}
and 
\begin{align}\label{eq:defAj}
\setA_j&=\Bigg\{\big(\matA,\matB,\matX\big) : \inf \Bigg\{\frac{\lVert\tp{(\tp{\veca_1}\matU\vecb_1\ \dots\ \tp{\veca_k}\matU\vecb_k)}\rVert_2}{\lVert\matU\rVert_2^{1/\beta}}: \matU\in\setU_\matX\!\setminus\!\{\matzero\} \Bigg\}>\tfrac{1}{j}\Bigg\},
\end{align}
for all $j\in\naturals$, where 
\begin{align}
\setU_\matX=\{\matU-\matX:\matU\in\setU\}, \quad\text{for all $\matX\in\reals^{m\times n}$.} 
\end{align}
By the same arguments as used in \cite[Section VI]{striagbo15}, one can show that $\setA$ is a measurable set. 
Application of Fubini's Theorem \cite[Theorem 1.14]{ma99} therefore yields 
\begin{align}\label{eq:allpyFubini}
\int \opP[(\matA,\matB,\rmatX)\in\setA]\,\mathrm d \lebmeasure(\matA,\matB) 
&= \opE[ \lebmeasure\{(\matA,\matB): (\matA,\matB,\rmatX)\in\setA\} ].
\end{align}
As the right-hand side of \cref{eq:allpyFubini} equals zero owing to  \cref{prp:nsreg}, it follows that 
\begin{align}\label{eq:PropA}
 \opP[(\matA,\matB,\rmatX)\in\setA]=0, \quad\text{for Lebesgue a.a. $(\matA,\matB)$.}
\end{align} 
Since the complement of $\setA$, denoted by $\setA^\mathrm{c}$, can be written as 
\begin{align}
\setA^\mathrm{c}=\bigcup_{j\in\naturals}\setA_j,
\end{align}
application of \hspace{1sp}\cite[Lemma 3.4, Item (a)]{ba95} together with \cref{eq:PropA} yields
\begin{align}\label{eq:prob1}
\lim_{j\to \infty} \opP[(\matA,\matB,\rmatX)\in\setA_j]=1, \quad\text{for Lebesgue a.a. $(\matA,\matB)$.}
\end{align}
Let $\setC\subseteq \reals^{m\times k}\times  \reals^{n\times k}$ denote the set of matrices $(\matA,\matB)$ for which \cref{eq:prob1} holds, and fix $\varepsilon>0$ arbitrarily.  
Then, for every $(\matA,\matB)\in\setC$, there must exist a $J(\matA,\matB)\in\naturals$ such that 
\begin{align}\label{eq:PropAJ}
 \opP[(\matA,\matB,\rmatX)\in\setA_{J(\matA,\matB)}]\geq 1-\varepsilon.
\end{align}
Next, for every $(\matA,\matB)\in\setC$,  let 
\begin{align}
\setU_{\matA,\matB}=\{\matX\in\setU:(\matA,\matB,\matX)\in\setA_{J(\matA,\matB)}\}. 
\end{align}
Since $\opP[\rmatX\in\setU]=1$ by assumption, \cref{eq:PropAJ} yields  
\begin{align}\label{eq:PropAU}
 \opP[\rmatX\in\setU_{\matA,\matB}]\geq 1-\varepsilon, \quad\text{for all $(\matA,\matB)\in\setC$.}
\end{align}
Now, consider $(\matA,\matB)\in\setC$ and fix  $\matU,\matV\in\setU_{\matA,\matB}$  with  $\matU\neq \matV$ but arbitrary otherwise. 
It follows  that  $\matU-\matV\in\setU_\matV\!\setminus\!\{\matzero\}$ and $(\matA,\matB,\matV)\in\setA_{J(\matA,\matB)}$, and    
the definition of $\setA_{J(\matA,\matB)}$ (see \cref{eq:defAj}) yields
\begin{align}
\lVert\matU-\matV\rVert_2^\frac{1}{\beta}\leq J(\matA,\matB)\lVert\tp{(\tp{\veca_1}(\matU-\matV)\vecb_1\ \dots\ \tp{\veca_k}(\matU-\matV)\vecb_k)}\rVert_2. 
\end{align} 
By  \cite[Lemma 2]{striagbo15}, we can therefore conclude that, for  every $(\matA,\matB)\in\setC$, the mapping 
\begin{align}
f_{\matA,\matB}\colon\setU_{\matA,\matB}&\to \{\tp{(\tp{\veca_1}\matX\vecb_1\ \dots\ \tp{\veca_k}\matX\vecb_k)}:\matX\in \setU_{\matA,\matB}\} \\
\matX&\mapsto \tp{(\tp{\veca_1}\matX\vecb_1\ \dots\ \tp{\veca_k}\matX\vecb_k)}
\end{align}
is injective with $\beta$-H\"older continuous inverse $f_{\matA,\matB}^{-1}$. 
Finally, for  every $(\matA,\matB)\in\setC$, the mapping $f_{\matA,\matB}^{-1}$ can be extended to the desired $\beta$-H\"older continuous mapping $g$ 
on $\reals^k$ by \cite[Theorem 1, Item ii)]{mi70}. 
\qed
\section{Proof of \cref{prp:nsreg}}\label{nsregproof}
For every $j\in\naturals$, 
let  $\setA(j)$ and $\setB(j)$ be as in \cref{eq:Al} and \cref{eq:Bl}, respectively. 
By countable subadditivity of  Lebesgue measure, it suffices to show that 
\begin{align}\label{eq:proptoshowreg2}
\inf \Bigg\{\frac{\lVert\tp{(\tp{\veca_1}\matX\vecb_1\ \dots\ \tp{\veca_k}\matX\vecb_k)}\rVert_2}{\lVert\matX\rVert_2^{1/\beta}}: \matX\in\setU\!\setminus\!\{\matzero\} \Bigg\}>0,
\end{align}
for Lebesgue  a.a.  $\big((\veca_1 \dots \veca_k), (\vecb_1 \dots \vecb_k)\big)  \in\setA(j)\times\setB(j)$ and  
all $j\in\naturals$.
Owing to \cref{lem:probzero2reg} below, \cref{eq:proptoshowreg2} then holds, for all $j\in\naturals$, with probability 1 if the deterministic matrices  $\big((\veca_1 \dots \veca_k), (\vecb_1 \dots \vecb_k)\big)  \in\setA(j)\times\setB(j)$ are replaced by independent random matrices with columns $\rveca_i$, $i=1,\dots,k$, independent and uniformly distributed on $\setB_m(\veczero,j)$,  and columns $\rvecb_i$, $i=1,\dots,k$, independent and uniformly distributed on $\setB_n(\veczero,j)$. 
By countable subadditivity of  Lebesgue measure, this finally implies that \cref{eq:proptoshowreg} can be violated only on a set of Lebesgue measure zero, which finalizes the proof.
 
\begin{lemma}\label{lem:probzero2reg}
Let $s>0$ and take $\rmatA=(\rveca_1\ \dots\ \rveca_k)$ and $\rmatB=(\rvecb_1\ \dots\ \rvecb_k)$ to be  independent random matrices with columns $\rveca_i$, $i=1,\dots,k$, independent and uniformly distributed on $\setB_m(\veczero,s)$,  and columns $\rvecb_i$, $i=1,\dots,k$, independent and uniformly distributed on $\setB_n(\veczero,s)$. 
Consider a nonempty and bounded set $\setU\subseteq\reals^{m\times n}$, 
and suppose that there exists a $\beta\in (0,1)$ such that 
\begin{align}
\frac{\overline{\dim}_\mathrm{B}(\setU)}{k} < 1-\beta. 
\end{align}
Then, 
\begin{align}
\opP\Bigg[
\inf \Bigg\{\frac{\lVert\tp{(\tp{\rveca_1}\matX\rvecb_1\ \dots\ \tp{\rveca_k}\matX\rvecb_k)}\rVert_2}{\lVert\matX\rVert_2^{1/\beta}}: \matX\in\setU\!\setminus\!\{\matzero\} \Bigg\}>0 
\Bigg]=1. 
\end{align}
\end{lemma}
\begin{proof}
Since $\setU$ is bounded by assumption, there exists a $K>0$ such that 
\begin{align}\label{eq:defK}
\sigma_1(\matX)\leq K, \quad\text{for all $\matX\in\setU$.}
\end{align}
For every $j\in\naturals$, set
\begin{align}\label{eq:defUj}
\setU_j=\setU\!\setminus\!\setB_{m\times n}\big(\matzero,2^{-\beta j}\big). 
\end{align} 
Then, \cref{eq:defK} together with  \cref{eq:defUj}, upon using $\sigma_1(\matX) \geq \lVert\matX\rVert_{2}/\rank(\matX)$,
yields
\begin{align}\label{eq:boundsigma}
\frac{2^{-\beta j}}{\sqrt{m}}\leq \sigma_1(\matX)\leq K, \quad\text{for all $\matX\in\setU_j$ and $j\in\naturals$.}
\end{align}
By \cref{lem:lemballs} below, it is sufficient to show that  
\begin{align} \label{eq:xyxy1}
\opP\big[\exists J: \lVert\tp{(\tp{\rveca_1}\matX\rvecb_1\ \dots\ \tp{\rveca_k}\matX\rvecb_k)}\rVert_2\geq 2^{-j},
 \text{ for all } \matX\in\setU_j, j\geq J
\big]=1. 
\end{align}
This will be established by arguing as follows. Suppose we can prove that there exists a $J\in\naturals$ such that 
\begin{align}\label{eq:xyxy2}
\sum_{j=J}^\infty\opP\big[\exists \matX\in \setU_j:\lVert\tp{(\tp{\rveca_1}\matX\rvecb_1\ \dots\ \tp{\rveca_k}\matX\rvecb_k)}\rVert_2< 2^{-j} 
\big]<\infty.
\end{align}
Then, the Borel-Cantelli Lemma \cite[Theorem 2.3.1]{du10} implies 
\begin{align}
\opP\big[\exists \matX\in \setU_j:\lVert\tp{(\tp{\rveca_1}\matX\rvecb_1\ \dots\ \tp{\rveca_k}\matX\rvecb_k)}\rVert_2< 2^{-j},
\ \text{for infinitely many $j\in\naturals$}\big]=0, 
\end{align}
which, in turn, implies \cref{eq:xyxy1}. 

It remains to establish \cref{eq:xyxy2}, which will be effected through a covering argument. For every $j\in \naturals$, consider the covering ball center
$\matY_i^{(j)}\in\setU_j$ such that 
\begin{align}
\setU_j\subseteq\bigcup_{i=1}^{N_{\setU_j}(2^{-j})} \setB_{m\times n}\Big(\matY_i^{(j)},2^{-j}\Big). 
\end{align}
A union bound argument then yields
\begin{align}
&\opP\big[\exists \matX\in \setU_j:\lVert\tp{(\tp{\rveca_1}\matX\rvecb_1\ \dots\ \tp{\rveca_k}\matX\rvecb_k)}\rVert_2< 2^{-j} \big]\label{eq:zuu1}\\
&\leq \sum_{i=1}^{N_{\setU_j}(2^{-j})}\opP\Big[\exists \matX\in \setB_{m\times n}\big(\matY_i^{(j)},2^{-j}\big):\lVert\tp{(\tp{\rveca_1}\matX\rvecb_1\ \dots\ \tp{\rveca_k}\matX\rvecb_k)}\rVert_2< 2^{-j} \Big].\label{eq:useunion}
\end{align}
Next, choose $J_1\in\naturals$ such that
\begin{align}
2^{-J_1(1-\beta)}\leq \frac{s^2}{(1+s^2\sqrt{k})\sqrt{m}}.\label{eq:J1def} 
\end{align} 
This implies $(1+s^2\sqrt{k})2^{-j}\leq \frac{2^{-\beta j}}{\sqrt{m}} s^2$, for all  $j\geq J_1$, 
and thus, by \cref{eq:boundsigma}, $(1+s^2\sqrt{k})2^{-j}\leq \sigma_1(\matX)s^2$, for all $\matX\in\setU_j$.
Hence, for all $j\geq J_1$, 
we can bound each summand in \cref{eq:useunion} according to
\begin{align}
& \opP\Big[\exists \matX\in \setB_{m\times n}\big(\matY_i^{(j)},2^{-j}\big):\lVert\tp{(\tp{\rveca_1}\matX\rvecb_1\ \dots\ \tp{\rveca_k}\matX\rvecb_k)}\rVert_2< 2^{-j} \Big]
\label{eq:boundptj1}
\\
&\leq \opP\Bigg[\Bigg\lVert\tp{\Big(\tp{\rveca_1}\matY^{(j)}_i\rvecb_1\ \dots\ \tp{\rveca_k}\matY^{(j)}_i\rvecb_k\Big)}\Bigg\rVert_2< (1+s^2\sqrt{k})2^{-j} \Bigg]\label{eq:usetriangle}
\\
&\leq (1+s^2\sqrt{k})^{k} 2^{-jk}\frac{2^{\frac{k(m+n)}{2}}}{\sigma_1\big(\matY^{(j)}_i\big)^k s^{2k} }\Bigg(1+\log \Bigg(\frac{s^2\sigma_1\big(\matY^{(j)}_i\big)}{(1+s^2\sqrt{k})2^{-j}}\Bigg)\Bigg)^k
\label{eq:applycom} 
\\
&\leq (s^{-2}+\sqrt{k})^{k} m^{\frac{k}{2}} 2^{-jk(1-\beta)}2^{\frac{k(m+n)}{2}}\Bigg(1+\log\Bigg( \frac{s^2K}{1+s^2\sqrt{k}}\Bigg)+j\log 2\Bigg)^k,
\label{eq:boundptj4}
\end{align}
where \cref{eq:usetriangle} is by \cref{eq:boundLS1}--\cref{eq:boundLS} for $\varepsilon_i=2^{-j}$, 
in \cref{eq:applycom} we applied \cref{lem:com1} with $\delta=(1+s^2\sqrt{k})2^{-j}$  and $\matX = \matY^{(j)}_i$,
and in \cref{eq:boundptj4} we used \cref{eq:boundsigma}. 
Inserting \cref{eq:boundptj1}--\cref{eq:boundptj4} into \cref{eq:zuu1}--\cref{eq:useunion} results in 
\begin{align}
&\opP\big[\exists \matX\in \setU_j:\lVert\tp{(\tp{\rveca_1}\matX\rvecb_1\ \dots\ \tp{\rveca_k}\matX\rvecb_k)}\rVert_2< 2^{-j} \big]\label{eq:use11}
\\
&\leq C N_\setU(2^{-j}) 
2^{-jk(1-\beta)} (D+j\log 2)^k, \quad\text{for all $j\geq J_1$},  \label{eq:use12}
\end{align}
with 
\begin{align}
C&=(s^{-2}+\sqrt{k})^{k} m^{\frac{k}{2}}  2^{\frac{k(m+n)}{2}}
\end{align}
and
\begin{align} 
D&= 1+\log \Bigg(\frac{s^2K}{1+s^2\sqrt{k}}\Bigg).
\end{align}
Next, let 
\begin{align}
d=\frac{\overline{\dim}_\mathrm{B}(\setU) +k(1-\beta)}{2},  
\end{align}
which implies  $\overline{\dim}_\mathrm{B}(\setU)<d<k(1-\beta)$ (see \cref{eq:assdimB}). 
By \cref{eq:dimbupper} we have
\begin{align}
\overline{\dim}_\mathrm{B}(\setU) = \inf_{\ell\in\naturals}\sup_{j\geq \ell}\frac{\log N_\setU\big(2^{-j}\big)}{\log\big(2^{j}\big)}.
\end{align}
Thus, as a consequence of $d>\overline{\dim}_\mathrm{B}(\setU)$, there exists a $J_2\in\naturals$
such that 
\begin{align}
N_\setU\big(2^{-j}\big) \leq 2^{jd}, \quad\text{for all $j\geq J_2$}. \label{eq:boundmink}
\end{align}
Now set $J=\max(J_1,J_2)$. Then, we have 
\begin{align}
&\sum_{j=J}^\infty \opP\big[\exists \matX\in \setU_j:\lVert\tp{(\tp{\rveca_1}\matX\rvecb_1\ \dots\ \tp{\rveca_k}\matX\rvecb_k)}\rVert_2< 2^{-j} \big]\\
&\leq C \sum_{j=J}^\infty N_\setU(2^{-j}) 2^{-jk(1-\beta)}(D+j\log 2)^k\label{eq:usebound}\\
&\leq C \sum_{j=J}^\infty 2^{-j(k(1-\beta)-d)} (D+j\log 2)^k\label{eq:usebound2}\\
&<\infty,\label{eq:usebound3}
\end{align}
where in \cref{eq:usebound} we used \cref{eq:use11}--\cref{eq:use12},  \cref{eq:usebound2}  is by \cref{eq:boundmink}, and 
\cref{eq:usebound3} follows from $d<k(1-\beta)$. 
\end{proof}

\begin{lemma}\label{lem:lemballs}
Consider a nonempty and bounded set $\setU\subseteq\reals^{m\times n} \setminus \{\matzero\}$ and let $f\colon \setU\to\reals^k$.   
Fix $\beta\in (0,1)$, and suppose that there exists a $J\in\naturals$ such that 
\begin{align}
\lVert f(\matX)\rVert_2\geq 2^{-j}, \quad\text{for all $\matX\in\setU\!\setminus\!\setB_{m\times n}\big(\matzero,2^{-\beta j}\big)$ and $j\geq J$.}
\end{align} 
Then, we have 
\begin{align}
\inf \Bigg\{\frac{\lVert f(\matX)\rVert_2}{\lVert\matX\rVert_2^{1/\beta}}: \matX\in\setU \Bigg\} >0. 
\end{align}
\end{lemma}
\begin{proof}
Follows from \cite[Lemma 3]{striagbo15} through vectorization. 
\end{proof}

\appendix
\section{Proof of \cref{LemmaExarec2}}\label{ProofLemmaExarec2}

\Cref{exarec} follows from \cite[Lemma III.1, Item i)]{albole19} through vectorization. 

In order to prove \cref{exasum2}, we first note that the sets $\setU_i$ participating in $\setU$ are all $s$-rectifiable by \cref{exarec}. 
To see that a finite union of $s$-rectifiable sets is $s$-rectifiable, we first prove
the statement for two sets and then note that the generalization to finitely many sets follows by induction.
Let  $\setA$ and  $\setB$ be $s$-rectifiable.
By the  definition of  rectifiability,  
there exist compact sets  $\setC, \setD\subseteq\reals^s$  and  Lipschitz mappings $\varphi\colon\setC\to\reals^{m\times n}$ and  $\psi\colon\setD\to\reals^{m\times n}$ such that 
$\setA=\varphi(\setC)$ and $\setB=\psi(\setD)$. 
As the sets  $\setC$ and $\setD$ are compact, there exists a constant $R>0$ such that $\setC \cup \setD \subseteq \setB_s(\veczero,R)$.
The set $\setD+\{3R\}$ is thus disjoint from $\setC$. 
We now define the function
\begin{align*}
  \tilde{\varphi} \colon 
  \setC \cup (\setD+\{3R\}) & \to \reals^{m\times n}
  \\
  \vecx & \mapsto 
    \begin{cases}
      \varphi(\vecx), & \vecx\in \setC \\
      \psi(\vecx-3R), & \vecx \in \setD+\{3R\}
    \end{cases}.
\end{align*}
The set $\setC \cup (\setD+\{3R\}) \subseteq\reals^s$ is compact as the union of compact sets
and $\tilde{\varphi}\big(\setC \cup (\setD+\{3R\})\big) = \setA\cup \setB$.
It remains to establish that $\tilde{\varphi}$ is Lipschitz. 
Indeed, for  vectors $\vecx, \vecy\in \setC$ the Lipschitz property follows from the Lipschitz property of $\varphi$.
Analogously, for $\vecx, \vecy\in \setD+\{3R\}$ the Lipschitz property is inherited from that of $\psi$.
For $\vecx\in \setC$ and $\vecy\in \setD+\{3R\}$, we have that $\lVert\vecx - \vecy\rVert \geq R$ and 
$\lVert\tilde{\varphi}(\vecx) - \tilde{\varphi}(\vecy)\rVert\leq 2\max_{\vecz\in \setC \cup (\setD+\{3R\})} \lVert\tilde{\varphi}(\vecz)\rVert=:M$.
Thus, 
$
\lVert\tilde{\varphi}(\vecx) - \tilde{\varphi}(\vecy)\rVert
\leq 
\frac{M}{R}\lVert\vecx - \vecy\rVert
$ and we obtain Lipschitz continuity of $\tilde{\varphi}$ with Lipschitz constant given by the maximum of $\frac{M}{R}$ and the Lipschitz constants of $\varphi$ and $\psi$.

To prove \cref{exaProdrec2}, let $\setU \in \reals^{m_1\times n_1}$ be $s$-rectifiable and $\setV \in \reals^{m_2\times n_2}$ $t$-rectifiable. 
By the  definition of  rectifiability, there exist compact sets  $\setC\subseteq\reals^s$ and $\setD\subseteq\reals^t$  and  Lipschitz mappings $\varphi\colon\setC\to\reals^{m_1\times n_1}$ and  
$\psi\colon\setD\to\reals^{m_2\times n_2}$ such that 
$\setU=\varphi(\setC)$ and $\setV=\psi(\setD)$. We can therefore write  
$\setU\times \setV=(\varphi\times \psi)(\setC\times \setD)$ with $\setC\times \setD\subseteq\reals^{s+t}$ compact and $\varphi\times \psi\colon \setC\times \setD\to\reals^{m_1\times n_1} \times \reals^{m_2\times n_2}$ Lipschitz. 

It remains to establish \cref{exaC12}. Let $\setK$ be a  compact subset of an $s$-dimensional $C^1$-submanifold  $\setM\subseteq\reals^{m\times n}$. The statement is trivial if 
$\setK=\emptyset$. We hence assume that  $\setK$ is nonempty.  
 By \cite[Definition 5.3.1]{krpa08}, we can write 
\begin{align}
\setM=\bigcup_{\matX\in\setM}\varphi_\matX(\setU_\matX), 
\end{align}
where, for every $\matX\in\setM$, $\setU_\matX\subseteq\reals^s$ is open, and $\varphi_\matX\colon\setU_\matX\to \reals^{m\times n}$ is a one-to-one $C^1$-map satisfying 
$\matX\in\varphi_\matX(\setU_\matX)$ and 
$\varphi_\matX(\setU_\matX)=\setV_\matX\cap\setM$ with $\setV_\matX\subseteq\reals^{m\times n}$ open. 
Since there exists a real analytic diffeomorphism between $\reals^s$ and $\setB_s(\veczero,1)$ \cite[Lemma K.10]{albole19}, we  can assume, w.l.o.g., that the sets $\setU_\matX$ are all bounded.  
 As $\setK\subseteq\setM$ is compact by assumption, 
 there must exist a finite  set $\{\matX_i: i=1,\dots, N\}\subseteq \setM$
such that 
\begin{align}
\setK\subseteq \bigcup_{i=1}^N \varphi_{\matX_i}(\setU_{\matX_i}) 
\end{align}
and $\setV_{\matX_i} \cap\setK\neq\emptyset$, for $i=1,\dots, N$. 
With the set $\{\matX_i: i=1,\dots, N\}\subseteq \setM$, we can now write  
\begin{align}
\setK
&=\bigcup_{i=1}^{N}(\varphi_{\matX_i}(\setU_{\matX_i})\cap\setK)\\
&=\bigcup_{i=1}^{N}\varphi_{i}(\setU_{i})\label{eq:stepa1}\\
&=\bigcup_{i=1}^{N}\varphi_{i}(\overline{\setU}_{i}),\label{eq:stepa2}   
\end{align}
where in \cref{eq:stepa1} we set $\varphi_i=\varphi_{\matX_i}$ and  $\setU_{i}=\setU_{\matX_i}\cap\varphi_i^{-1}(\setK)$, and  \cref{eq:stepa2} is by $\setK=\overline{\setK}$ and the continuity of 
$\varphi_i$. 
The claim now follows from \cref{exasum2} applied to \ref{eq:stepa2}.


\bibliographystyle{emss.bst}
\bibliography{references_RKSB}









\end{document}